\documentclass[11pt,a4paper,leqno]{article}
\usepackage[nocolor]{mymacros}
\usepackage{cite}

\title{Log-Sobolev inequality for the $\varphi^4_2$ and $\varphi^4_3$ measures}

\author{Roland Bauerschmidt\footnote{University of Cambridge, Statistical Laboratory, DPMMS. E-mail: {\tt rb812@cam.ac.uk}.}
\and
Benoit Dagallier\footnote{University of Cambridge, Statistical Laboratory, DPMMS. E-mail: {\tt bd444@cam.ac.uk}.}}

\date{\vspace*{-2em}} 

\begin{document}
\maketitle

\begin{abstract}
The continuum $\varphi^4_2$ and $\varphi^4_3$ measures are shown to satisfy a log-Sobolev
inequality uniformly in the lattice regularisation under the optimal assumption that their
susceptibility is bounded.
In particular, this applies to all coupling constants in any finite volume,
and uniformly in the volume
in the entire high temperature phases of the $\varphi^4_2$ and $\varphi^4_3$ models.

The proof uses a general criterion for the log-Sobolev inequality in terms of the
Polchinski (renormalisation group) equation, a recently proved remarkable
correlation inequality for Ising models with general external fields, the
Perron--Frobenius theorem, and bounds on the susceptibilities of the
$\varphi^4_2$ and $\varphi^4_3$ measures obtained using skeleton inequalities.
\end{abstract}

\section{The $\varphi^4_d$ measure}
Let $d=2$ or $d=3$, and
let $\Lambda_{\epsilon,L} = L\T^d\cap \epsilon\Z^d$ (and always assume $L$ is a multiple of $\epsilon$).
For $\lambda>0$ and $\mu \in \R$, the lattice regularised (continuum) $\varphi^4_d$ measure 
is defined as the probability measure
\begin{equation} 
  \nu^{\epsilon,L}_{\lambda,\mu}(d\varphi) \propto
  \exp\left[-\epsilon^d\sum_{x\in\Lambda_{\epsilon,L}} \left[ \frac12 \varphi_x(-\Delta^\epsilon\varphi)_x +\frac{\lambda}{4}\varphi_x^4 + \frac{\mu+a^\epsilon(\lambda)}{2} \varphi_x^2\right]\right]\, d\varphi,
  \label{eq_def_continuum_phi4_measure}
\end{equation}
where $d\varphi$ denotes the Lebesgue measure on $\R^{\Lambda_{\epsilon,L}}$,
the lattice Laplacian $\Delta^\epsilon$ is:
\begin{equation}
\forall \varphi\in\R^{\Lambda_{\epsilon,L}},
\qquad
(\Delta^\epsilon\varphi)_x =\epsilon^{-2}\sum_{y\sim x}\big[\varphi_y-\varphi_x\big],\label{eq_def_continuous_Laplacian}
\end{equation}
and $a^\epsilon(\lambda)$ is a dimension-dependent divergent counterterm which ensures that the $\epsilon\rightarrow 0$ limit of the measure
(on a suitable space of generalised functions on $L\T^d$) exists and is non-Gaussian.
The construction of the limiting measures has a long history, and we give some references further below.
The division of $\mu+a^\epsilon(\lambda)$ into a finite mass term $\mu$ (which plays the role of a temperature of the model)
and a divergent counterterm $a^\epsilon(\lambda)$ is only determined up to an additive bounded constant.
Explicitly, for an arbitrary fixed $m^2 > 0$, which we will usually take to be $m^2=1$, one can take $a^\epsilon(\lambda)=a^\epsilon(\lambda,m^2)$ with
\begin{equation} \label{e:counterterm1}
  a^\epsilon(\lambda,m^2) := -3\lambda \big(-\Delta^\epsilon+m^2)^{-1}(0,0) + 6\lambda^2 \big\|\big(-\Delta^\epsilon+m^2\big)^{-1}(0,\cdot)\big\|^3_{L^3}.
\end{equation}
In this definition and subsequently,
the matrix elements of $(-\Delta^\epsilon+m^2)^{-1}$ are normalised 
with respect to the inner product
$(u,v)_\epsilon = \epsilon^d \sum_{x\in\Lambda_{\epsilon,L}} u(x)v(x)$
so that $(-\Delta^\epsilon+m^2)^{-1}(x,y)$ converges to its continuum counterpart,
and $L^p$ norms are defined by $\|f\|_{L^p}^p = \|f\|_{L^p(\Lambda_{\epsilon,L})}^p = \epsilon^d \sum_{x\in\Lambda_{\epsilon,L}} |f(x)|^p$.
This implies the following scaling 
(with dimension-dependent constants $c_i>0$):
\begin{align}  \label{e:counterterm2}
  a^\epsilon(\lambda,m^2)  
  = 
  \begin{cases}
    -c_1\lambda \log(\epsilon^{-2}) + O_{m^2,\lambda}(1)  & (d=2),\\
    -c_1\lambda\epsilon^{-1}+c_2\lambda^2\log(\epsilon^{-2}) + O_{m^2,\lambda}(1) & (d=3).
  \end{cases} 
\end{align}
Note that when $d=2$ the $\lambda^2$ counterterm in \eqref{e:counterterm1} is bounded
and part of $O_{m^2,\lambda}(1)$ in \eqref{e:counterterm2},
and could thus be dropped, but we will keep it to unify the presentation of both dimensions.
Moreover, as discussed already,
any value of $m^2>0$ can be used,
and 
for this reason, we often do not explicitly mention $m^2$ in the notation $\nu^{\epsilon,L}_{\lambda,\mu}$.
Different choices of $m^2$ only correspond to a different choice of origin for $\mu$.
Expectation with respect to \eqref{eq_def_continuum_phi4_measure} 
is denoted by $\big<\cdot\big>^{\epsilon,L}_{\lambda,\mu}$.

The Glauber dynamics with invariant measure $\nu^{\epsilon,L}_{\lambda,\mu}$, also known as the dynamical $\varphi^4_d$ model, is the solution to the systems of SDEs
\begin{equation} \label{e:SDE}
  d\varphi_t = (\Delta^\epsilon \varphi_t - \lambda \varphi_t^3 - (\mu+a^\epsilon(\lambda))\varphi_t)\, dt + \sqrt{2}dW_t^{\epsilon,L}
\end{equation}
where $dW_t^{\epsilon,L}$ is space-time white noise on $\R_+ \times \Lambda_{\epsilon,L}$, i.e., $W^{\epsilon,L}$ is a Brownian motion on $L^2(\Lambda_{\epsilon,L})$
(or, in other words, the $W_t^{\epsilon,L}(x)$ are independent Brownian motions of variance $\epsilon^{-d}$ for $x\in \Lambda_{\epsilon,L}$).
For $d=3$, the pathwise existence of the $\epsilon\to 0$ limit is a main result of Hairer's theory of regularity structures
\cite{MR3785597,MR3274562,MR3468250}; for $d=2$ this existence is simpler \cite{MR2016604}; for $d=3$ see also \cite{MR3406823,MR3846835,MR3758734}
and \cite{MR3459120} for alternative approaches to the pathwise limit.
By standard results for SDEs, the solution to \eqref{e:SDE} determines a Markov process with Dirichlet form given by
\begin{equation} \label{e:dirichlet_continuum}
  D(F) = \epsilon^{-d} \sum_{x\in\Lambda_{\epsilon,L}} \avgB{ \Big(\frac{\partial F}{\partial\varphi_x}\Big)^2 }_{\lambda,\mu}^{\epsilon,L}.
\end{equation}
Our main result is a log-Sobolev inequality for the $\varphi^4_d$ measures, stated in Theorem~\ref{thm:lsiphi4d} below.
We formulate it for
$\epsilon>0$ with uniform $\epsilon$ dependence
(and also uniform $L$ dependence when the susceptiblity is bounded).
All constants that appear in its statement could be made explicit.
The log-Sobolev constant of \eqref{eq_def_continuum_phi4_measure} is the largest constant $\gamma = \gamma^{\epsilon,L}(\lambda,\mu)$ such that
\begin{equation} \label{e:lsi_continuum}
\forall F:\R^{\Lambda_{\epsilon,L}}\rightarrow\R_+,\qquad  
\ent(F) \leq \frac{2}{\gamma} D(\sqrt{F}),
\end{equation}
where, for $F$ nonnegative, the relative entropy is defined by
\begin{equation}
  \ent(F) = \avg{F \log F}_{\lambda,\mu}^{\epsilon,L} - \avg{F}_{\lambda,\mu}^{\epsilon,L}\log \avg{F}_{\lambda,\mu}^{\epsilon,L}.
\end{equation}
The log-Sobolev inequality has many general implications,
see the references~\cite{MR3155209,MR1971582} for reviews of these.
In particular, it is equivalent to the hypercontractivity of the Markov semigroup associated with the Dirichlet form,
and to its exponential relaxation to equilibrium in relative entropy sense for initial conditions of finite relative entropy.

For the continuum $\varphi^4_d$ measures, the log-Sobolev inequality was conjectured in~\cite{MR1370101}. 
Spectral gap inequalities 
(quantifying the relaxation in $L^2$ norm with respect to the invariant measure) and exponential ergodicity (in total variation)
for the $\varphi^4_d$ measures in finite volume have been proven in \cite{MR3825880} for $d=2$
and extended to $d=3$ in \cite{10.1017/fmp.2021.18}.
These results apply directly in the continuum  limit $\epsilon\to 0$.
The Dirichlet form in the continuum limit was identified in \cite{MR3858906}.
A main ingredient of the proof of the above spectral gap inequalities is a qualitative support theorem for the continuum SPDEs and a compactness argument.
In particular, the spectral gap obtained this way is not uniform in $L$ and not explicit, and the validity of the log-Sobolev inequality also remained open.

We remark that, for any $\epsilon>0$ and $L<\infty$, the regularised $\varphi^4_d$ measure is convex at infinity and
standard techniques (see, e.g., \cite{MR1837286})
show that $\gamma^{\epsilon,L}(\lambda,\mu)>0$, but the bound obtained in this way
tends to $0$ (extremely  quickly) as $\epsilon\to 0$ or $L\to\infty$,
as the counterterms \eqref{e:counterterm1}--\eqref{e:counterterm2}
make the microscopic measure very non-convex
as $\epsilon \to 0$ (for any $\lambda>0$ and $\mu \in \R$).

\begin{theorem} \label{thm:lsiphi4d}
  Let $d=2$ or $d=3$ and $\epsilon>0$. Let $L \geq 1$ assumed to be a multiple of $\epsilon>0$.
  
  (i)
  Let $\lambda>0$, $\mu \in \R$,
  and assume that there is a constant $\bar\chi \in (0,\infty)$ such that
  \begin{equation} \label{e:chi-ub}
    \chi^{\epsilon,L}(\lambda,\mu) :=
    \epsilon^d \sum_{x\in\Lambda_{\epsilon,L}} \langle \varphi_0\varphi_x\rangle^{\epsilon,L}_{\lambda,\mu} \leq \bar\chi.
  \end{equation}
  Then the log-Sobolev constant $\gamma^{\epsilon,L}(\lambda,\mu)$ of \eqref{eq_def_continuum_phi4_measure}
  is bounded below by a positive constant $\bar\gamma = \bar\gamma(\lambda,\mu,\bar\chi)$ uniformly in $\epsilon$
  and the bound depends only on $(\lambda,\mu,\bar\chi)$ and not directly on $L$ (or $\epsilon$):
  \begin{equation}
    \gamma^{\epsilon,L}(\lambda,\mu) \geq \bar \gamma(\lambda,\mu,\bar\chi).
  \end{equation}
  
  (ii)
  For any $\lambda,\mu>0$, there are $\mu(d,\lambda),\lambda(d,\mu)>0$ such that
  if the counterterms in \eqref{e:counterterm1}--\eqref{e:counterterm2} are chosen with $m^2=\mu$ instead of $m^2=1$,
  and if either $\mu>\mu(d,\lambda)$ or $\lambda \in[0,\lambda(d,\mu)]$, 
  then the log-Sobolev constant satisfies,
  uniformly in $\epsilon$ and $L$,
  \begin{equation} \label{e:mugamma-(ii)}
    C^{-1} \mu \leq \gamma^{\epsilon,L}(\lambda,\mu) \leq C \mu.
  \end{equation}
\end{theorem} 

The choice $m^2=\mu$ instead of $m^2=1$ in item (ii) amounts to a different choice of origin for $\mu$,
  and \eqref{e:mugamma-(ii)} states that this choice ensures that the log-Sobolev constant is comparable to the mass term $\mu$.
Equivalently, this statement could be formulated in terms of $m^2=1$ when replacing $\mu$ in \eqref{e:mugamma-(ii)} on both sides by $\mu+a^{\epsilon}(\lambda,\mu)-a^{\epsilon}(\lambda,1)$, where $a^\epsilon(\lambda,\mu)-a^\epsilon(\lambda,\mu)$ converges as $\epsilon \to 0$.

Item (i) says that the log-Sobolev constant is bounded below as soon as the susceptibility is bounded above. 
This condition is optimal, in the sense that the converse implication also holds, i.e., the susceptibility is bounded above if the log-Sobolev constant is bounded below.
Indeed, by using the trial function $F(\varphi) = \epsilon^d L^{-d/2} \sum_{x\in\Lambda_{\epsilon,L}} \varphi_x$ and observing that
$\var(F) = \chi^{\epsilon,L}(\lambda,\mu)$ and $D(F) = 1$, the general fact that
the spectral gap is bounded below by the log-Sobolev constant implies that the optimal log-Sobolev constant $\gamma^{\epsilon,L}(\lambda,\mu)$
of  \eqref{eq_def_continuum_phi4_measure} satisfies
\begin{equation}
  \gamma^{\epsilon,L}(\lambda,\mu) \leq \frac{1}{\chi^{\epsilon,L}(\lambda,\mu)}.
\end{equation}
Together with Theorem~\ref{thm:lsiphi4d} (i)  this shows that $\gamma^{\epsilon,L}(\lambda,\mu)$ is bounded below if and only if $\chi^{\epsilon,L}(\lambda,\mu)$ is bounded above.

The boundedness of the susceptibility $\chi^{\epsilon,L}(\lambda,\mu)$ in both $\epsilon$ and $L$ provides a standard definition of the
critical temperature $\mu_c(\lambda)$. Thus, in other words, Theorem~\ref{thm:lsiphi4d} (i) shows that the log-Sobolev constant
of the $\varphi^4_d$ measure is bounded below throughout the high temperature phase $\mu>\mu_c(\lambda)$.
For background, we summarise the following known properties of the susceptibility of $\varphi^4_d$ models ($d=2,3$),
say with the counterterms \eqref{e:counterterm1} defined with $m^2=1$:
\begin{enumerate}
\item For any $\lambda>0$ and $\mu\in \R$, the finite volume susceptibility $\chi^{\epsilon,L}(\lambda,\mu)$ is bounded in $\epsilon$ for any fixed $L<\infty$.
  For example, this can be shown using the methods of the recent works \cite{MR4252872,MR4173157}.
\item For any $\lambda>0$, there is $\mu_0(\lambda) \in \R$ such that for $\mu>\mu_0(\lambda)$, the susceptibility is also bounded both in $\epsilon$ and $L$.
  For example, a very simple proof of this was given in \cite{MR723546} whose method of skeleton inequalities we also apply as an ingredient;
  see Section~\ref{sec:suscept-phi4d}.
\item For any $\lambda>0$, there is $\mu_1(\lambda) \in \R$ such that for $\mu<\mu_1(\lambda)$, the susceptibility is not bounded in $L$.
  For $d=2$, see \cite{MR0443695}, and for $d=3$, see \cite{MR0421531} and the recent work \cite{MR4426324}. 
\item
  By the second Griffiths inequality,
  the susceptibility is decreasing in $\mu$ for any fixed $\lambda>0$, provided the counterterms are defined in terms of $m^2$ independent of $\mu$, e.g., $m^2=1$.
  Thus with this choice (ii) and (iii) determine a unique $\mu_c(\lambda) \in \R$ that separates the $\mu$ for which the susceptibility is bounded or unbounded.
\end{enumerate}

While the above shows that the log-Sobolev constant is positive throughout the high temperature phase,
the bound obtained as $\mu \downarrow \mu_c(\lambda)$ is explicit in terms of the susceptibility but far from the expected truth.
For large $\mu$, item~(ii) of Theorem~\ref{thm:lsiphi4d} determines the log-Sobolev constant up to multiplicative constants.

For general background on the $\varphi^4_2$ measures, see the textbooks \cite{MR0489552,MR887102},
and for $\varphi^4_3$
the discussions of the existing references in the introductions of \cite{MR723546,MR4252872}.
Aside from the references already mentioned above,
recent progress of the SPDE perspective include \cite{MR3693966,MR3719541,MR4164267}.

Our proof of Theorem~\ref{thm:lsiphi4d} relies on four main ingredients: the criterion for the log-Sobolev inequality
in terms of the Polchinski (renormalisation group) equation from \cite{MR4303014},
the recently proved correlation inequality for the Ising model with general external fields from \cite{2107.09243},
the Perron--Frobenius theorem,
and improvements of the estimates for the $\varphi^4_d$ measures from \cite{MR723546}.
As a simple application of the first part of the proof, we also show (the also new result) that the log-Sobolev constant of the lattice $\varphi^4$ model
is bounded below uniformly in the volume up to the critical point (see Section~\ref{sec:latticephi4}).
In \cite{lsiising}, we apply related methods to the Ising model.

\section{Criterion for the log-Sobolev inequality}
\label{sec:criterion}

In this section, we formulate a sufficient criterion for a log-Sobolev inequality for $\varphi^4$ measures.
The criterion is first stated for $\varphi^4$ measures in a general setting (thus applying in particular to any graph and any dimension),
in Theorem~\ref{thm:criterion},  and
then specialised to the continuum $\varphi^4_d$ measure \eqref{eq_def_continuum_phi4_measure} in Theorem \ref{thm:criterion_continuum}.
The criterion could be further extended to more general spin systems in the Griffiths--Simon class \cite{MR428998}.
We do not pursue  this generalisation explicitly, but
in \cite{lsiising}, we present a variant of the criterion for the Ising model (with an additional ingredient needed to pass from continuous to discrete spins).

Let $\Lambda$ be a finite set, let $A = (A_{x,y})_{x,y\in\Lambda}\subset \R$ a symmetric matrix,
let $g>0$ and $\nu \in\R$,
and consider the general $\varphi^4$ measure 
with external field $h\in\R^\Lambda$ on $\Lambda$:
\begin{equation}
  \mu^{\Lambda,h}_{A,g,\nu}(d\varphi) \propto
  \exp\Big[-\frac{1}{2}\big(\varphi,A\varphi\big)- V_0(\varphi)+ (h,\varphi)\Big] \, d\varphi,\label{eq_def_lattice_phi4_measure}
\end{equation}
with $(u,v) = \sum_{x\in\Lambda}u_xv_x$ for $u,v\in\R^\Lambda$ (there is no $\epsilon$ here),
and $V_0$ the potential:
\begin{equation} \label{eq_V0general}
\forall \varphi\in\R^\Lambda,\qquad V_0(\varphi) =  \sum_{x\in\Lambda} (\frac14 g\varphi^4_x + \frac12 \nu\varphi^2_x).
\end{equation}
Since the set $\Lambda$ is fixed in the following, we omit it and write $\mu^{h}_{A,g,\nu}$ instead of $\mu^{\Lambda,h}_{A,g,\nu}$, 
and $\big<\cdot\big>^{h}_{A,g,\nu}$ for the associated expectation. 
Also, when $h=0$, we simply write $\mu_{A,g,\nu}$ and $\big<\cdot \big>_{A,g,\nu}$.
In this setting, 
the log-Sobolev constant of \eqref{eq_def_lattice_phi4_measure}   is the best constant $\gamma=\gamma(A,g,\nu)$ such that:
\begin{equation}
 \forall F:\R^\Lambda\rightarrow\R_+,\qquad 
 \ent(F) \leq \frac{2}{\gamma} D(\sqrt{F}),
\end{equation}
with:
\begin{equation}
  \ent(F) = \avg{F \log F}_{A,g,\nu} - \avg{F}_{A,g,\nu}\log \avg{F}_{A,g,\nu},
    \qquad 
    D(F) = \sum_{x\in\Lambda}\Big<\Big(\frac{\partial F}{\partial \varphi_x}\Big)^2\Big>_{A,g,\nu}.
\end{equation}
For $t>0$, consider the measure $\mu_{A,g,\nu+1/t}$ with additional mass $1/t$, 
and write $S_t = S_{A,g,\nu+1/t}$ for its two-point correlation function:
\begin{equation}
\forall x,y\in\Lambda,\qquad S_t(x,y) := \big<\varphi_x\varphi_y\big>_{A,g,\nu+1/t}.
\end{equation}
Introduce also the susceptibility $\chi_t = \chi_{A,g,\nu+1/t}$:
\begin{equation}
\chi_t := \max_{x\in\Lambda}\sum_{y\in\Lambda}S_t(x,y).\label{eq_susceptibility}
\end{equation}
In particular, when $\Lambda$ is a torus and $A_{x,y} = A_{0,x-y}$ for all $x,y \in\Lambda$, then
$S_t$ is translation invariant as well, and $\chi_t$ is the usual susceptibility.
Define then:
\begin{equation}
  \dot\kappa_t = \frac{1}{t} - \frac{\chi_t}{t^2}.\label{eq_dot_lambda_t}
\end{equation}
\begin{theorem} \label{thm:criterion}
Assume that $A$ is positive definite, 
and that offdiagonal coefficients of $A$ are nonpositive. 
The log-Sobolev constant is then bounded by:
\begin{equation}
  \frac{1}{\gamma}\leq \int_0^\infty\exp\Big[-2\int_0^t\dot\kappa_s \, ds\Big]\, dt.
\end{equation}
\end{theorem}

\begin{remark}
  The proof only uses that
  $\chi_t$ is an upper bound on the  spectral radius of the covariance matrix $(S_t(x,y))_{x,y\in\Lambda}$
  and one could therefore instead define $\chi_t$ in this way for a slightly more general statement.
\end{remark}

The proof of Theorem \ref{thm:criterion} has several ingredients.
The starting point is the general criterion for the log-Sobolev inequality derived from the Polchinski (renormalisation group) equation
in \cite{MR4303014}.
Let us recall this general criterion.
In \cite[Section~2]{MR4303014}, this criterion is formulated for a general covariance decomposition $(C_t)$.
We will make the particular choice (known as Pauli--Villars regularisation):
\begin{equation}
  C_t = (A+ 1/t)^{-1} \quad (t>0), \qquad C_0 = 0,
\end{equation}
and recall that the matrix $A$ is positive definite. 
We write $\dot C_t,\ddot C_t$ for the first and second derivatives of $C_t$. 
Let ${\bf E}_{C_t}$ denote expectation under the Gaussian measure with covariance $C_t$, 
and introduce a renormalised potential $V_t$ by:
\begin{equation} \label{e:Vt-def}
  \forall \varphi\in\R^\Lambda,\qquad V_t(\varphi) = -\log {\bf E}_{C_t}\Big[e^{-V_0(\varphi+\cdot)}\Big].
\end{equation}
Equivalently, $V_t$ satisfies the Polchinski (renormalisation group) equation
\begin{equation} \label{e:polchinski}
  \ddp{V_t}{t} = \frac12 \Delta_{\dot C_t} V_t - \frac12 (\nabla V_t)^2_{\dot C_t},
\end{equation}
see \cite{MR4303014} for notation and discussion of this equation,
and \cite{polchinski1984269} for its original use. 
The Polchinski equation is not used directly in this
paper, but it underlies the following proposition. 
\begin{proposition}[Theorem 2.5 in \cite{MR4303014}] \label{thm:polchinski}
The log-Sobolev constant $\gamma=\gamma(A,g,\nu)$ for the measure $\mu_{A,g,\nu}$ is bounded in terms of $V_t$ as follows:
\begin{equation}
\frac{1}{\gamma}\leq \int_0^\infty\exp\Big[-2\int_0^t\dot\ell_s\, ds\Big]\, dt,
\end{equation}
where, for $t>0$, $\dot\ell_t$ is any real number such that, 
for all $\varphi \in \R^\Lambda$:
\begin{equation}
\dot C_t\He V_t(\varphi)\dot C_t - \frac{1}{2}\ddot{C}_t\geq \dot\ell_t\dot C_t.\label{eq_def_dot_l}
\end{equation}
\end{proposition}

Thus we need to show that $\dot\ell_t$ may be chosen to be $\dot\kappa_t$.
To this end, observe that $V_t$ defined in \eqref{e:Vt-def} can be expressed,
by changing variables   from $\zeta$ to $\zeta-\varphi$, as
\begin{align}
  V_t(\varphi) + \text{const}
  &= -\log \int e^{-\frac12 (\zeta,C_t^{-1}\zeta)} e^{-V_0(\varphi+\zeta)} \, d\zeta
    \nnb
  &= \frac12 (\varphi,C_t^{-1}\varphi)-\log \int e^{-\frac12 (\zeta,C_t^{-1}\zeta)} e^{-V_0(\zeta)} e^{(C_t^{-1}\varphi,\zeta)}\, d\zeta,
\end{align}
and hence the Hessian of $V_t$ is given by:
\begin{equation} 
  \He V_t(\varphi)
  = C_t^{-1} - C_t^{-1}\Sigma_t(\varphi)C_t^{-1}
  ,  
  \label{eq_Hessian_V_t}
\end{equation}
where
\begin{equation} \label{e:Sigma-phi}
  \Sigma_t(\varphi) = \Big(\big<\zeta_x;\zeta_y \big>^{C_t^{-1}\varphi}_{A,g,\nu+1/t}\Big)_{x,y \in\Lambda},
\end{equation}
with the expectation $\big<\cdot\big>$ here acting on the field $\zeta$ and the following notation for the covariance:
\begin{equation}
\big<F;G\big>^h_{A,g,\nu} 
= \big<FG\big>^h_{A,g,\nu} - \big<F\big>^h_{A,g,\nu} \big<G\big>^h_{A,g,\nu},
\qquad 
F,G:\R^\Lambda\rightarrow\R,\quad h\in\R^\Lambda.
\end{equation}

The second ingredient of the proof of Theorem~\ref{thm:criterion} is the
remarkable correlation inequality
\cite[Corollary~1.3]{2107.09243}
which shows that
the truncated two-point function of ferromagnetic Ising models with general external field
is maximised pointwise when the external field is $0$,
recently proved by Ding--Song--Sun.
By the Griffiths--Simon construction \cite{MR428998}, this inequality extends to $\varphi^4$ models as follows 
(the proof of this extension is exactly as the extension of the other correlation inequalities from Ising models to $\varphi^4$ models in \cite{MR428998}).

\begin{proposition}\label{prop_correlation_Ding_Song_Sun}
  Let $g>0$, $\nu\in\R$, and let $A$  be a matrix with nonpositive offdiagonal coefficients. For any $h \in\R^\Lambda$, then
  \begin{equation}
\forall x,y\in\Lambda,\qquad    
0 \leq \big<\varphi_x;\varphi_y\big>^{h}_{A,g,\nu}
    \leq \big<\varphi_x;\varphi_y\big>^{0}_{A,g,\nu} = \big<\varphi_x\varphi_y\big>^{0}_{A,g,\nu}.
  \end{equation}  
\end{proposition}

To use this inequality to bound the Hessian in \eqref{eq_Hessian_V_t} requires a third additional ingredient,
as a quadratic form instead of pointwise bound is needed.
This ingredient in the proof of Theorem~\ref{thm:criterion} is the Perron--Frobenius theorem,
which enables us to bound the eigenvalues of the Hessian of $V_t$ in \eqref{eq_Hessian_V_t} through Proposition \ref{prop_correlation_Ding_Song_Sun}.
The combination of the last two ingredients culminates in the following proposition
from which Theorem~\ref{thm:criterion} follows immediately.

\begin{proposition}\label{prop_Perron_Frobenius}
For each $t> 0$, one can take $\dot\ell_t = \dot\kappa_t$ in \eqref{eq_def_dot_l},
where  $\dot\kappa_t$ is defined by \eqref{eq_dot_lambda_t}.
\end{proposition}
\begin{proof}[Proof of Proposition \ref{prop_Perron_Frobenius}]
Fix $t>0$, take $X\in\R^\Lambda$, 
and using the expression  \eqref{eq_Hessian_V_t} for the Hessian notice that
\begin{align}
X^T\He V_t(\varphi)X 
  \geq
  X^TC_t^{-1}X - \|C_t^{-1}X\|_2^2  \sup_{\|Y\|_2=1}Y^T\Sigma_t(\varphi)Y.\label{eq_bound_hessian_0}
\end{align}
The correlation matrix $\Sigma_t(\varphi)$ has nonnegative entries, 
as correlations under $\mu^h_{A,g,\nu}$ are nonnegative for any $h\in \R^\Lambda$ by the FKG inequality.
By the Perron--Frobenius theorem, the spectral radius of the correlation matrix is therefore an eigenvalue, 
and there is an associated (normalised) eigenvector $Y_0=Y_0(\varphi)$ with nonnegative entries. 
The pointwise correlation bound of Proposition~\ref{prop_correlation_Ding_Song_Sun} 
thus implies:
\begin{equation}
  \sup_{\|Y\|_2=1}Y^T\Sigma_t(\varphi)Y=Y^T_0\Sigma_t(\varphi)Y_0 \leq  Y^T_0\Sigma_t(0)Y_0 \leq \|\Sigma_t(0)\|
  = \sup_{\|Y\|_2=1} Y^T \Sigma_t(0)Y
  .
\end{equation}
The spectral radius $\|\Sigma\|$ of a matrix $\Sigma$ is a lower bound for any matrix norm. 
For a matrix $\Sigma$ with nonnegative entries, this implies:
\begin{equation}
\|\Sigma\|\leq \max_{x}\sum_{y}\Sigma_{x,y}.\label{eq_Hoelder_bound_spectral_radius}
\end{equation}
For $\Sigma=\Sigma_t(0)$ given by \eqref{e:Sigma-phi},
the right-hand side of \eqref{eq_Hoelder_bound_spectral_radius} is precisely the susceptibility~\eqref{eq_susceptibility}. 
This and \eqref{eq_bound_hessian_0} imply the following bound on the lowest eigenvalue of the Hessian of $V_t$:
\begin{align}
X^T\text{Hess}(V_t)(\varphi)X 
\geq 
X^T\big(C_t^{-1} - \chi_tC_t^{-2}\big)X.\label{eq_bound_Hessian_by_susceptibility}
\end{align}
To conclude the proof of Proposition \ref{prop_Perron_Frobenius}, 
notice that, for each $t>0$:
\begin{align}
\dot C_t = \frac{1}{t^2}(A+1/t)^{-2}  = \frac{C^2_t}{t^2},
\qquad 
\ddot C_t &= -\frac{2}{t^3}A(A+1/t)^{-3}\nonumber\\
 &= -\frac{2}{t}AC_t\dot C_t.
\end{align}
Injecting these expressions in the criterion \eqref{eq_def_dot_l} for the  log-Sobolev inequality
concludes the proof of Proposition~\ref{prop_Perron_Frobenius} ($C_t,\dot C_t,\ddot C_t$, and $A$ all commute):
\begin{align}
\dot C_t\text{Hess}(V_t)(\varphi)\dot C_t - \frac{1}{2}\ddot{C}_t
&\geq\Big[
\frac{C_t}{t} \big(C_t^{-1} - \chi_tC_t^{-2}\big)\frac{C_t}{t} +\frac{1}{t}AC_t 
\Big]\dot C_t\nonumber\\
&= \Big[\frac{1}{t}\Big(\frac{1}{t}+A\Big)C_t - \frac{\chi_t}{t^2} 
\Big]\dot C_t \nonumber\\
&= \Big[\frac{1}{t}- \frac{\chi_t}{t^2}\Big]\dot C_t
= \dot \kappa_t \dot C_t.
\end{align}
\end{proof}

\begin{proof}[Proof of Theorem~\ref{thm:criterion}]
  As already mentioned, the proof of the theorem is immediate from
  Propositions~\ref{thm:polchinski} and~\ref{prop_Perron_Frobenius}.
\end{proof}

Finally, we consider the special choice of $\Lambda,A,g,\nu$ corresponding to the continuum $\varphi^4_d$ measure $\nu^{\epsilon,L}_{\lambda,\mu}$
defined in \eqref{eq_def_continuum_phi4_measure}, with an arbitrary $m^2>0$ to define the counterterms \eqref{e:counterterm1}.
To emphasise the value of $m^2>0$ used to define the counterterms,
we write $\nu^{\epsilon,L}_{\lambda,\mu,m^2}$ and $\big<\cdot\big>^{\epsilon,L}_{\lambda,\mu,m^2}$
for the measure and its expectation.
This choice is:
\begin{equation}
\Lambda = \Lambda_{\epsilon,L},\quad 
A = \epsilon^d\big(-\Delta^\epsilon + m^2\big),\quad 
g = \epsilon^d\lambda,\quad
\nu = \epsilon^d(\mu-m^2 + a^\epsilon(\lambda,m^2)).\label{eq_choice_parameters_continuum_phi4}
\end{equation}
The factors of $\epsilon^d$ result from
the fact that we did not include factors $\epsilon^d$ in \eqref{eq_V0general}
and the use of the standard inner product $(u,v) = \sum_x u_xv_x$ in this section.
In the application to the continuum model,
it makes sense to normalise the scale parameter $t$ in continuum scaling also,
i.e., to rescale $t$ by $\epsilon^d$.
This means that we consider the measures $\nu^{\epsilon,L}_{\lambda,\mu+1/t,m^2}$ 
and the correspondingly normalised susceptibility:
\begin{equation}
\chi^{\epsilon,L}_t(\lambda,\mu,m^2) 
:= \max_{x\in\Lambda_{\epsilon,L}}\epsilon^d\sum_{y\in\Lambda_{\epsilon,L}}\big<\varphi_x\varphi_y\big>^{\epsilon,L}_{\lambda,\mu+1/t,m^2}
= \epsilon^d\sum_{x\in\Lambda_{\epsilon,L}}\big<\varphi_0\varphi_x\big>^{\epsilon,L}_{\lambda,\mu+1/t,m^2},
\end{equation}
where the last equality is by translation invariance, and we set:
\begin{equation}
\dot\kappa^{\epsilon,L}_t := \frac{1}{t} - \frac{\chi^{\epsilon,L}_t(\lambda,\mu,m^2)}{t^2}.
\end{equation}
Theorem \ref{thm:criterion} then takes the following form. The log-Sobolev constant in its statement is
normalised as in \eqref{e:lsi_continuum} with respect to the continuum normalised Dirichlet form \eqref{e:dirichlet_continuum}.
Note that $A$ in \eqref{eq_choice_parameters_continuum_phi4} is indeed positive definite with nonpositive offdiagonal entries. 
\begin{theorem}\label{thm:criterion_continuum}
Let $\lambda>0$ and $\mu\in\R$.  
The log-Sobolev constant $\gamma^{\epsilon,L}(\lambda,\mu,m^2)$ of \eqref{eq_def_continuum_phi4_measure}
normalised as in \eqref{e:lsi_continuum} and
with counterterms defined in terms of mass $m^2>0$
satisfies the bound:
\begin{equation}
  \frac{1}{\gamma^{\epsilon,L}(\lambda,\mu,m^2)} 
  \leq 
  \int_0^\infty\exp\Big[-2\int_0^t\dot\kappa^{\epsilon,L}_s\, ds\Big]\, dt. 
\end{equation}
\end{theorem}

%

\section{Simple application: Log-Sobolev inequality for lattice $\varphi^4$ models}
\label{sec:latticephi4}

To illustrate the criterion for the log-Sobolev inequality from Theorem~\ref{thm:criterion},
we give a simple proof that the log-Sobolev constant of the lattice $\varphi^4$ spin model in any dimension is bounded
uniformly in the volume throughout its entire high temperature phase.
For unbounded spin systems including the lattice $\varphi^4$ measure,
the existence of a spectral gap uniformly in the volume was previously known only away from the critical temperature
\cite{MR1715549,MR1704666,MR1837286}.

Fix a dimension $d\geq 1$. For simplicity of the exposition, we assume that $\Lambda \subset \Z^d$ is a hypercube with periodic boundary conditions, i.e., a
discrete torus, although Dirichlet boundary conditions could be considered with essentially no change.
For $g>0$, $\nu\in\R$,
recall that the lattice $\varphi^4$ model on $\Lambda$ is defined by
\begin{equation} \label{e:lattice-phi4-measure}
  \mu^\Lambda_{g,\nu}(d\varphi) = e^{-\frac12(\varphi,-\Delta\varphi)-\sum_{x\in\Lambda} (\frac14 g \varphi^4 + \nu\varphi^2)},
\end{equation}
where $\Delta$ is the lattice Laplacian on $\Lambda$, and its susceptibility is defined by
\begin{equation}
  \chi^\Lambda(g,\nu) = \sum_{x\in\Lambda} \avg{\varphi_0\varphi_x}_{g,\nu}.
\end{equation}
The critical value of the $d$-dimensional lattice $\varphi^4$ model can be defined by
\begin{equation}
  \nu_c(g) = \inf \{\nu \in\R: \sup_\Lambda \chi^\Lambda(g,\nu) < +\infty\},
\end{equation}
where the supremum is over all $d$-dimensional discrete tori.
In fact, $\nu_c(g)<0$ for $g>0$.

\begin{example}
The log-Sobolev constant $\gamma = \gamma^\Lambda_{g,\nu}$ of the measure \eqref{e:lattice-phi4-measure} satisfies
\begin{equation}
  \frac{1}{\gamma} \leq \frac{e^2}{2|\nu|+1} + (2|\nu|+1)^3 e^{2+2(2|\nu|+1)\chi^\Lambda(g,\nu)}
  < \infty.
\end{equation}
\end{example}

Thus for any $g>0$ and $\nu \in\R$ such that $\chi^\Lambda(g,\nu)$ is bounded uniformly in $\Lambda$, the log-Sobolev constant $\gamma$
is bounded below uniformly in $\Lambda$. Since on the other hand the lower boundedness of the log-Sobolev constant also implies
that the susceptibility is bounded
(analogously to the discussion below Theorem~\ref{thm:lsiphi4d}),
it thus follows that the log-Sobolev constant is bounded below
if and only if the susceptibility is bounded.

\begin{proof}
Note that the interesting case is $\nu\leq 0$; otherwise the potential is convex and one could use the Bakry--Emery criterion
to show the log-Sobolev inequality.
In bounding the susceptibility $\chi_t$, let us separately treat the small and large $t$ cases as follows.

In the small $t$ case given by $t \leq 1/(2|\nu|+1)$, the measure $\mu_{g,\nu+1/t}$ is of the form $e^{-U(\varphi)}\, d\varphi$,
where $U = \frac12 (\varphi,(-\Delta+\nu+1/t) \varphi) + \text{convex}$ is a strictly convex potential, with Hessian bounded below by $(1/t +\nu)\id>0\id$.  
Therefore the Brascamp--Lieb inequality \cite[Theorem~4.1]{BraLieAppl} implies
\begin{equation} \label{e:BLappl}
\chi_t 
= \var_{\mu_{g,\nu+1/t}^\Lambda} \Big( (\varphi,{\bf 1})\Big)
\leq \frac{1}{\frac{1}{t}+\nu},
\qquad
{\bf 1}:= \frac{1}{|\Lambda|^{1/2}}(1,\dots,1).
\end{equation}
Recalling that $\dot\kappa_t = 1/t - \chi_t/t^2$ is the quantity appearing in the bound of the log-Sobolev constant in Theorem~\ref{thm:criterion}, 
one has:
\begin{equation}
\forall t\leq 1/(2|\nu|+1),\qquad 
\dot\kappa_t\geq \frac{1}{t}\Big(1-\frac{1}{1+t\nu}\Big) 
= \frac{\nu}{1+t\nu}.\label{eq_small_t_lattice_phi4}
\end{equation}
Consider now the large $t$ case where $t>1/(2|\nu|+1)$. Since $\chi^\Lambda(g,\nu)$ is decreasing in $\nu \in \R$ by the second Griffiths inequality,  
\begin{equation}
\forall t>0,\qquad  
\dot\kappa_t \geq \frac{1}{t} - \frac{\chi^\Lambda(g,\nu+1/t)}{t^2}
  \geq \frac{1}{t} - \frac{\chi^\Lambda(g,\nu)}{t^2}.
\end{equation}
Using this bound on $\dot\kappa_t$ for $t > 1/(2|\nu|+1)$ and the small $t$ bound \eqref{eq_small_t_lattice_phi4} for $t\leq 1/(2|\nu|+1)$:
\begin{align}
\forall t> 0,\qquad  
\kappa_t := \int_0^t\dot\kappa_s\, ds 
&\geq 
\int_0^{t\wedge 1/(2|\nu|+1)} (-2|\nu|) \, dt + \int_{1/(2|\nu|+1)}^{t\vee 1/(2|\nu|+1)} (\frac{1}{t}-\frac{\chi^\Lambda(g,\nu)}{t^2}) \, dt
              \nnb
  &\geq -1 + {\bf 1}_{t> 1/(2|\nu|+1)}\Big[\log(t/(2|\nu|+1)) - (2|\nu|+1)\chi^\Lambda(g,\nu)\Big]. 
\end{align}
This implies:
\begin{equation}
\forall t>0,\qquad  
e^{-2\kappa_t} 
  \leq e^2\bigg[
  {\bf 1}_{t\leq 1/(2|\nu|+1)} + \frac{(2|\nu|+1)^2}{t^2} e^{2(2|\nu|+1)\chi^\Lambda(g,\nu)} {\bf 1}_{t> 1/(2|\nu|+1)}
  \bigg].
\end{equation}
Therefore, by Theorem~\ref{thm:criterion}, 
the log-Sobolev constant satisfies
\begin{equation}
  \frac{1}{\gamma_\Lambda} \leq \frac{e^2}{2|\nu|+1} + (2|\nu|+1)^3 e^{2+2(2|\nu|+1)\chi^\Lambda(g,\nu)}
\end{equation}
as needed.
\end{proof}

\section{Bound on the susceptibility of the continuum $\varphi^4_d$ model} 
\label{sec:suscept-phi4d}
In dimensions $d \in \{2,3\}$,
consider the continuum $\varphi^4_d$ measure \eqref{eq_def_continuum_phi4_measure}
with parameters $\lambda>0$ and $\mu\in\R$. 
In this section, we will make the dependence on the mass $m^2$ of the counterterms \eqref{e:counterterm1} explicit.
The measure \eqref{eq_def_continuum_phi4_measure} is thus denoted by $\nu^{\epsilon,L}_{\lambda,\mu,m^2}$, 
and the associated expectation by $\big<\cdot\big>^{\epsilon,L}_{\lambda,\mu,m^2}$. 

\subsection{Statement of the result}

In order to apply Theorem~\ref{thm:criterion_continuum}, which will be done in Section~\ref{sec:proofs},
our goal in this section is to estimate, for $t> 0$, 
the susceptibility $\chi^{\epsilon,L}_t(\lambda,\mu,m^2)$ of the $\varphi^4_d$ measure with mass $\mu+1/t$ (and counterterm defined in terms of $m^2>0$), defined by
\begin{equation}
\chi^{\epsilon,L}_t(\lambda,\mu,m^2) := \epsilon^d \sum_{x\in\Lambda_{\epsilon,L}}\big<\varphi_0\varphi_x\big>^{\epsilon,L}_{\lambda,\mu+1/t,m^2}.
\end{equation}
The parameter $t$ is the scale parameter
in the Polchinski (renormalisation group) equation  \eqref{e:polchinski}. 
Informally, we speak of large scale when $t$ is large, and of small scale when $t$ is small. 

For $t>0$, further introduce the notation:
\begin{gather}
S_t(x) := \big<\varphi_0\varphi_x\big>^{\epsilon,L}_{\lambda,\mu+1/t,m^2},
\qquad x\in\Lambda_{\epsilon,L},
\\
  C_t(x) := \Big(-\Delta^\epsilon+m^2+\frac{1}{t}\Big)^{-1}(0,x),
  \qquad x\in\Lambda_{\epsilon,L},
\end{gather}
and note that both functions are positive (e.g., by the FKG inequality).
Define also the shorthand:
\begin{equation}
m^2_t := m^2 +\frac{1}{t},
\end{equation}
and recall the definition of the norms $\|\cdot\|_{L^p},\|\cdot\|_{L^\infty}$ norms for $p\geq 1$:
\begin{equation}
\|f\|_{L^p} := \Big(\epsilon^d\sum_{x\in\Lambda_{\epsilon,L}}|f(x)|^p\Big)^{1/p},
\qquad
\|f\|_{L^\infty} := \max_{x\in\Lambda_{\epsilon,L}}|f(x)|,
\qquad f\in\R^{\Lambda_{\epsilon,L}}.
\end{equation}
Finally, introduce for $t>0$ the differences $\eta_t,\gamma_t$
between counterterms defined at masses $m^2$ and $m^2_t$ 
(recall the definition \eqref{e:counterterm1} of the counterterm $a^\epsilon(\lambda,m^2)$):
\begin{equation}
  \eta_t = C_\infty(0)-C_t(0) \geq 0,
  \qquad
  \gamma_t = \|C_\infty^3\|_{L^1} - \|C_t^3\|_{L^1} \geq 0.
\end{equation}

In the following $c>0$ denotes a constant that is
independent from the parameters $\epsilon,L,\lambda,\mu,m^2$ of the model
and may change from line to line. 
To express dependence on a parameter, say $\mu$, we write $c(\mu)$. 

Since $\chi_t^{\epsilon,L}(\lambda,\mu,m^2) = \|S_t\|_{L^1}$ and $\|C_t\|_{L^1} = m_t^{-2}$ 
the desired estimates on the susceptibility
$\chi_t^{\epsilon,L}(\lambda,\mu,m^2)$ will be seen to be a direct consequence of the following propositions.

\begin{proposition}[Small scale, all couplings]\label{prop_bound_susceptibility}
Let $d\in\{2,3\}$, $\lambda>0$, $\mu\in\R$, and $m^2>0$.  
Then there is $t_0 = t_0(d,\lambda,\mu,m^2)>0$ and a polynomial $p_{d,t,\mu,m^2}(\lambda)$ in $\lambda$ of degree independent of the parameters, 
with coefficients functions of $\eta_t,\gamma_t,\mu,m^2$ times positive powers of $m_t^{-1}$ (in particular independent of $\epsilon,L$), such that:
\begin{equation}
\forall t\in(0,t_0],\qquad
\|S_t-C_t\|_{L^1}
\leq 
p_{d,t,\mu,m}(\lambda),
\end{equation}
and $t\mapsto t^{-2} p_{d,t,\mu,m}(\lambda)$ is integrable on $(0,t_0]$.
\end{proposition}

When $\mu>0$, one can make the special choice $m^2=\mu$ as in the statement of Theorem~\ref{thm:lsiphi4d} (ii).
In this case, the results of Proposition~\ref{prop_bound_susceptibility} are valid for large $t$ (i.e., large scale) provided either $\lambda$ is small enough, or $\mu=m^2$ is large enough.
\begin{proposition}[All scales, one small coupling]\label{prop_bound_susceptibility_long_time}
Let $d\in\{2,3\}$, $\lambda>0$, $\mu>0$, and set $m^2:=\mu$. 
There are then $\lambda(d,\mu),\mu(d,\lambda)>0$ and a polynomial $\tilde p_{d,t,\mu}(\lambda)$ in $\lambda$, 
again of degree independent of the parameters and 
with coefficients functions of $\eta_t,\gamma_t,\mu$ times positive powers of $m_t^{-1}$ (but independent of $\epsilon,L$), 
such that the following holds: 
if either $\lambda\in[0,\lambda(d,\mu))]$ or $\mu\geq \mu(d,\lambda)$, then:
\begin{equation}
\forall t>0,
\qquad \|S_t-C_t\|_{L^1}
\leq 
\tilde p_{d,t,\mu}(\lambda),
\end{equation}
the function $t\mapsto t^{-2}\tilde p_{d,t,\mu}(\lambda)$ is integrable on $(0,\infty)$, 
and the integral is independent of $\lambda$ and $\mu$ if either $\lambda\in[0,\lambda(d,\mu)]$ or $\mu\geq \mu(d,\lambda)$.
Under the same conditions, 
\begin{equation}
\limsup_{t\to\infty }\tilde p_{d,t,\mu}(\lambda) \underset{\mu\rightarrow\infty}{=} o(\mu^{-1}).
\end{equation}
\end{proposition}

\subsection{Skeleton inequalities}

To prove Propositions \ref{prop_bound_susceptibility}--\ref{prop_bound_susceptibility_long_time}, 
the starting point is a bound on correlations presented in Proposition~\ref{prop_BFS} below, 
obtained by Brydges, Fr\"ohlich, and Sokal \cite{MR723546},
using the method of skeleton inequalities \cite{MR719815} (which is based on the random walk representation \cite{MR648362}).
To state it, let $\star$ denote the convolution:
\begin{equation}
(f\star g)(x) = \epsilon^d\sum_{y\in\Lambda_{\epsilon,L}}f(x-y)g(y),
\qquad f,g\in\R^{\Lambda_{\epsilon,L}}.
\end{equation}
For a function $f\in \R^{\Lambda_{\epsilon,L}}$, write also:
\begin{equation}
({\bf 1}^\epsilon_0  f)(x) := \epsilon^{-d}{\bf 1}_{x=0}f(0).
\end{equation}
The following inequalities are exact upper and lower bounds on the two-point function $S$ of the $\varphi^4_d$ model,
consistent with naive perturbation theory except that the right-hand sides also involve the
interacting two-point function (or propagator) $S$ rather than only the noninteracting one $C$. 

\begin{proposition}[Equations (5.12)--(5.13) in \cite{MR723546}]\label{prop_BFS}
Consider the $\varphi^4_d$ measure \eqref{eq_def_continuum_phi4_measure} with parameters $\mu+1/t\in\R$ (in place of $\mu$), $\lambda>0$, and $m^2>0$. 
Then, for each $x\in\Lambda_{\epsilon,L}$:
\begin{align}
S_t(x)-C_t(x) &\geq -3\lambda S_t(0)(C_t\star S_t)(x) + 6\lambda^2(C_t \star S_t^3\star S_t)(x) \nnb
&\quad -54\lambda^3 (C_t\star Q_t\star S_t)(x) - \big(a^\epsilon(\lambda,m^2)+\mu-m^2\big) (S_t\star C_t)(x),
\end{align}
and:
\begin{align}
S_t(x)-C_t(x) &\leq -3\lambda S_t(0)(C_t\star S_t)(x) + 6\lambda^2(C_t \star S_t^3\star S_t)(x) \nnb
&\quad - \big(a^\epsilon(\lambda,m^2)+\mu-m^2\big) (S_t\star C_t)(x),
\end{align}
where:
\begin{equation} 
Q_t := S_t\big(S_t^2\star S_t^2\big).\label{eq_def_Q}
\end{equation}
\end{proposition}

By writing the counterterms $a^\epsilon(\lambda,m^2)$ from \eqref{e:counterterm1} 
in terms of $C_\infty(x) = (-\Delta^\epsilon+m^2)^{-1}(0,x)$,
the last proposition implies:
\begin{align}
\big|S_t(x)-C_t(x)\big| &\leq 3\lambda|S_t(0)-C_\infty(0)| (C_t\star S_t)(x) 
+ 6\lambda^2\big|(C_t \star \big(S_t^3 - {\bf 1}^\epsilon_0  \|C_\infty^3\|_{L^1}\big) \star S_t)(x)\big| \nnb
&\quad + 
54\lambda^3 (C_t\star Q_t\star S_t)(x) + |\mu-m^2|\big(C_t\star S_t\big)(x).\label{eq_bound_BFS_3d}
\end{align}
Equation \eqref{eq_bound_BFS_3d} is the starting point for the proof of Propositions \ref{prop_bound_susceptibility}--\ref{prop_bound_susceptibility_long_time}, 
which takes the rest of Section \ref{sec:suscept-phi4d}. 
We use the short-hands:
\begin{equation}
C:= C_t,\quad S:= S_t,\quad Q:= Q_t,\quad E:= S-C.
\end{equation}
When using \eqref{eq_bound_BFS_3d} to prove Proposition \ref{prop_bound_susceptibility}--\ref{prop_bound_susceptibility_long_time}, 
there are two main differences to the computations done in \cite{MR723546}. 
The first (minor) difference is that we allow a possibly negative mass $\mu+1/t$, and the counterterms are defined at mass $m^2>0$. In contrast, the counterterms in \cite{MR723546} are defined at the same mass $m^2>0$ as the one defining the massive Gaussian measure to which the $\varphi^4$ measure is compared to. 
The second (more significant) difference is that we need to pay attention to the $t$-dependence in all bounds:
Firstly, the counterterms are defined at mass $m^2$, not $m^2_t = m^2+1/t$. Secondly, the resulting bounds on $\|S_t-C_t\|_{L^1}$ need to be integrable
with respect to $t^{-2} \, dt$
since we are after an estimate on the susceptibility $\|S_t\|_{L^1}$ in order to estimate the log-Sobolev constant using Theorem~\ref{thm:criterion_continuum}.  \\
  
  Due to the above observations, it is convenient to split the the differences $S(0)-C_\infty(0)$ and $S^3-{\bf 1}^\epsilon_0 \|C_\infty^3\|_{L^1}$ as follows:
\begin{align}
&S(0)-C_\infty(0) = S(0)-C(0) + C(0)-C_\infty(0) =: S(0)-C(0) - \eta_t\nnb
  &S^3-{\bf 1}^\epsilon_0\|C_\infty^3\|_{L^1} = \pb{S^3 - C^3} + \pb{C^3 - {\bf 1}^\epsilon_0  \|C^3\|_{L^1}}
    - {\bf 1}^\epsilon_0  \gamma_t,
    \label{eq_splitting_counterterms}
\end{align}
where we recall:
\begin{equation}
  \eta_t = C_\infty(0)-C(0) \geq 0,
  \qquad
  \gamma_t = \|C_\infty^3\|_{L^1} - \|C^3\|_{L^1} \geq 0.\label{eq_def_eta_gamma}
\end{equation}

Let us lastly remark that, while proving Proposition \ref{prop_bound_susceptibility}--\ref{prop_bound_susceptibility_long_time} ultimately requires  bounding $\|E\|_{L^1}$,
 due to the term $S(0)-C_\infty(0)$ in \eqref{eq_bound_BFS_3d}, 
$\|E\|_{L^1}$ can be bounded through \eqref{eq_bound_BFS_3d} only in terms of both $\|E\|_{L^1}$ and $\|E\|_{L^\infty}$. 
To prove Propositions \ref{prop_bound_susceptibility}--\ref{prop_bound_susceptibility_long_time}, we thus first estimate:
\begin{equation}
\|E\|_{L^1\cap L^\infty}:= \|E\|_{L^1}+\|E\|_{L^\infty}.
\end{equation}
The resulting bound on $\|E\|_{L^1 \cap L^\infty}$ of course implies a bound on $\|E\|_{L^1}$, 
but the $t$-dependence of this bound turns out to be neither optimal, nor sufficient for our purposes. 
However, using that $\|E\|_{L^\infty}$ is already controlled as input, we improve the bound by rerunning the argument in the $L^1$ norm only, 
which then yields the claim of both propositions.
\subsection{Estimate in $\|\cdot\|_{L^1\cap L^\infty}$ norm}
We separately compute the terms related to each power of $\lambda$ appearing in \eqref{eq_bound_BFS_3d}. 
The next lemma gathers useful properties and estimates. 
The bounds on moments of $C$ are proven in Appendix \ref{app_bounds_on_diagram}.
\begin{lemma}\label{lemm_bound_first_two_moments_C}
(i) For $t>0$, recall the shorthand $m^2_t:= m^2 +1/t$ and $C:= C_t$. 
One has:
\begin{equation}
\|C\|_{L^1} = \frac{1}{m^2_t},
\qquad 
\|C^2\|_{L^1} \leq \frac{c}{m^2_t}{\bf 1}_{d=2} + \frac{c}{m_t}{\bf 1}_{d=3}.
\end{equation}
(ii) If $f,g\in\R^{\Lambda_{\epsilon,L}}$, then:
\begin{equation}
\|f\star g\|_{L^1}\leq \|f\|_{L^1}\|g\|_{L^1},
\quad
\|f\star g\|_{L^\infty} \leq \|f\|_{L^2}\|g\|_{L^2},
\quad \|f\star g\|_{L^\infty} \leq \|f\|_{L^1}\|g\|_{L^\infty}.
\end{equation}
In fact $\|f\star g\|_{L^1} = \|f\|_{L^1}\|g\|_{L^1}$ when $f,g\geq 0$.
\end{lemma}
\subsubsection{Order 0 term}
The only term independent of $\lambda$ in the bound \eqref{eq_bound_BFS_3d} is:
\begin{equation}
T_{\lambda^0} := |\mu-m^2|\, C\star S
= |\mu-m^2| \big(C\star C + C\star E\big).
\end{equation}
Using Lemma~\ref{lemm_bound_first_two_moments_C}, 
its $\|\cdot \|_{L^1}$ and $\|\cdot \|_{L^\infty}$ norms satisfy:
\begin{align}
\|T_{\lambda^0}\|_{L^\infty}
&\leq 
|\mu-m^2| 
\Big(\|C\|_{L^1}\|E\|_{L^\infty} + \|C\|_{L^2}^2\Big)\nnb
&\leq 
|\mu-m^2|\Big(
\frac{\|E\|_{L^\infty}}{m^2_t} + \frac{c}{m^2_t}{\bf 1}_{d=2} + \frac{c}{m_t}{\bf 1}_{d=3}\Big),
\end{align}
and:
\begin{align}
\|T_{\lambda^0}\|_{L^1}
&\leq
|\mu-m^2| \Big(
\frac{\|E\|_{L^1}}{m^2_t} + \frac{1}{m^4_t}\Big).
\end{align}
As a result,
\begin{equation}
\|T_{\lambda^0}\|_{L^1\cap L^\infty} 
\leq 
|\mu-m^2|
	\Big[
	\frac{\|E\|_{L^1\cap L^\infty}}{m^2_t}
	 + \frac{1}{m^4_t}+\frac{c}{m^2_t}{\bf 1}_{d=2} + \frac{c}{m_t}{\bf 1}_{d=3} 
	 \Big].
\end{equation}
\subsubsection{Order 1 term}
The first order term is defined by:
\begin{equation}
T_{\lambda} :=
3\lambda \big[S(0)-C(0) + C(0)-C_\infty(0)\big] S\star C
= 3\lambda \big[S(0)-C(0) - \eta_t\big] S\star C,
\end{equation}
where $\eta_t$ is the difference \eqref{eq_def_eta_gamma} between counterterms with mass $m^2$ and $m^2_t$. 
It is computed in Lemma \ref{lemm_counterterm}. 
As for $T_{\lambda^0}$, 
the $\|\cdot\|_{L^\infty}$ norm of the first order term $T_{\lambda}$ reads:
\begin{align}
\big\|T_\lambda\|_{L^\infty} 
&\leq 
3\lambda\big(\|E\|_{L^\infty} + \eta_t\big)\Big(\|C\|_{L^1}\|E\|_{L^\infty} + \|C\|_{L^2}^2\Big)\nonumber\\
&\leq 
3\lambda \big(\eta_t + \|E\|_{L^\infty}\big)\Big(\frac{\|E\|_{L^\infty}}{m^2_t} + \frac{c}{m^2_t}{\bf 1}_{d=2} + \frac{c}{m_t}{\bf 1}_{d=3}\Big),
\end{align}
In contrast, its $\|\cdot\|_{L^1}$ norm is simply:
\begin{equation}
\|T_{\lambda}\|_{L^1}
\leq 
3\lambda\big(\eta_t + \|E\|_{L^\infty}\big)\|C\star S\|_{L^1}
\leq \frac{3\lambda\big(\eta_t + \|E\|_{L^\infty}\big)}{m^2_t}\Big(\|E\|_{L^1}+\frac{1}{m^2_t}\Big).\label{eq_L1_norm_lambda_term}
\end{equation}
We conclude on the first order term:
\begin{equation}
\|T_{\lambda}\|_{L^1\cap L^\infty}
\leq 
c\lambda\big(\eta_t + \|E\|_{L^1\cap L^\infty})
	\Big[
	\frac{\|E\|_{L^1\cap L^\infty}}{m^2_t}
	 + \frac{1}{m^4_t}+\frac{1}{m^2_t}{\bf 1}_{d=2} + \frac{1}{m_t}{\bf 1}_{d=3} 
	 \Big].
\end{equation}
\subsubsection{Order 2 terms}\label{sec_L1Linfty_bound_T2}
The second order contribution $T_{\lambda^2}$ corresponds to the following terms:
\begin{align}
T_{\lambda^2}
:= 
6\lambda^2 \Big[C\star \big(S^3-C^3\big)\star S 
 + C\star \big({\bf 1}^\epsilon_0\|C^3\|_{L^1} - {\bf 1}^\epsilon_0\|C_\infty^3\|_{L^1}\big)\star S 
 +  C\star \psi \star S\Big], 
\end{align}
where $\psi$ is the function:
\begin{equation}
\psi : = C^3-{\bf 1}^\epsilon_0\|C^3\|_{L^1}.
\end{equation}
To compute the term involving $S^3-C^3$, write:
\begin{align}
C \star \big(S^3 - C^3\big) \star S = C\star E\big(3C^2 + 3EC + E^2\big)\star(E+C). \label{eq_decomp_Scube_minus_C3}
\end{align}
The $\|\cdot\|_{L^\infty}$ norm of this term then reads:
\begin{align}
\big\|C \star \big(S^3 - C^3\big) \star S\|_{L^\infty}
&\leq 
3\|E\|_{L^\infty} \|C\star C^2\star C\|_{L^\infty}
+ 3\|E\|_{L^\infty}^2\Big(\|C\star C\star C\|_{L^\infty} + \|C\star C^2\|_{L^1}\Big) \nonumber\\
  &\qquad + 4\|E\|_{L^\infty}^3\|C\|_{L^1}^2 + \|E\|_{L^\infty}^3\|E\|_{L^1}\|C\|_{L^1}\label{eq_bound_Scube_minus_Ccube_term_3d}
    .
\end{align}
Cauchy-Schwarz inequality and Lemma~\ref{lemm_bound_first_two_moments_C} yield an estimate of the terms involving $C$:
\begin{align}
\|C\star C\star C\|_{L^\infty}
\leq 
\|C\|_{L^1}\|C\|^2_{L^2} = \frac{1}{m^2_t}\cdot&
\begin{cases}
\frac{c}{m^2_t}\quad &\text{if }d=2,\\
\frac{c}{m_t}\quad &\text{if }d=3,
\end{cases}
\\
\|C\star C^2\star C\|_{L^\infty}
\leq  
\|C\star C\|_{L^\infty} \|C^2\|_{L^1}
\leq  
\|C^2\|^2_{L^1}
\leq &\begin{cases}
\frac{c}{m^4_t}\quad &\text{if }d=2,\\
\frac{c}{m^2_t}\quad &\text{if }d=3,
\end{cases}
\\
\|C^2\star C\|_{L^1}
= 
\|C^2\|_{L^1}\|C\|_{L^1}
 \leq \frac{1}{m^2_t}\cdot&\begin{cases}
\frac{c}{m^2_t}\quad &\text{if }d=2,\\
\frac{c}{m_t}\quad &\text{if }d=3.
\end{cases}
\end{align}
It follows that \eqref{eq_bound_Scube_minus_Ccube_term_3d} becomes:
\begin{align}
\big\|C \star \big(S^3 - C^3\big) \star S\|_{L^\infty}
&\leq 
c\|E\|_{L^\infty}\Big(\frac{1}{m^4_t}{\bf 1}_{d=2}  + \frac{1}{m^2_t}{\bf 1}_{d=3}\Big)
+ c\|E\|_{L^\infty}^2\Big(\frac{1}{m^4_t}{\bf 1}_{d=2} 
	+ \frac{1}{m^3_t}{\bf 1}_{d=3}\Big) \nonumber\\
&\qquad 
+ \|E\|_{L^\infty}^3\Big(\frac{4}{m^4_t} + \frac{\|E\|_{L^1}}{m^2_t}\Big).
\end{align}
Similarly, we find for the $\|\cdot\|_{L^1}$ norm:
\begin{align}
&\|C\star (S^3-C^3)\star S\|_{L^1}
\leq 
\|C\star S\|_{L^1} \|E(3C^2+3EC+E^2\|_{L^1}\nonumber\\
&\qquad\leq 
\frac{1}{m^2_t}\Big(\frac{1}{m^2_t}+\|E\|_{L^1}\Big) 
\Big(3\|E\|_{L^\infty}\|C^2\|_{L^1} + 3\|E\|^2_{L^\infty}\|C\|_{L^1} + \|E\|_{L^\infty}^2\|E\|_{L^1}\Big)\nonumber\\
&\qquad\leq 
\frac{1}{m^2_t}\Big(\frac{1}{m^2_t}+\|E\|_{L^1}\Big) 
\Big[c\|E\|_{L^\infty}\Big(\frac{1}{m^2_t}{\bf 1}_{d=2} + \frac{1}{m_t}{\bf 1}_{d=3}\Big)
	+ \|E\|_{L^\infty}^2\Big(\frac{3}{m^2_t} + \|E\|_{L^1}\Big)\Big].\label{eq_L1_norm_first_lambda2_term}
\end{align}
We conclude on the $\|\cdot\|_{L^1\cap L^\infty}$ norm of $C\star (S^3-C^3)\star S$:
\begin{align}
\|C\star (S^3-C^3)\star S\|_{L^1\cap L^\infty}
&\leq 
c\|E\|_{L^1\cap L^\infty}\Big[
\Big(\frac{1}{m^4_t}+ \frac{1}{m^6_t}\Big) {\bf 1}_{d=2}+ \Big(\frac{1}{m^2_t} + \frac{1}{m^5_t}\Big){\bf 1}_{d=3} 
\Big]\nonumber\\
&\quad + c\|E\|_{L^1\cap L^\infty}^2
	\Big(\frac{1}{m^6_t} + \frac{1}{m^4_t}{\bf 1}_{d=2} + \frac{1}{m^3_t}{\bf 1}_{d=3}\Big)\nonumber\\
&\quad +\frac{c\|E\|_{L^1\cap L^\infty}^3}{m^4_t}
 + \frac{c\|E\|^4_{L^1\cap L^\infty}}{m^2_t}.
\end{align}
Another $\lambda^2$ term is $C\star({\bf 1}^\epsilon_0  \|C^3\|_{L^1} - {\bf 1}^\epsilon_0  \|C_\infty^3\|_{L^1})\star S$,
which simply reads:
\begin{align}
C\star({\bf 1}^\epsilon_0  \|C^3\|_{L^1} - {\bf 1}^\epsilon_0  \|C_\infty^3\|_{L^1})\star S 
&= \big(\|C^3\|_{L^1} - \|C^3_\infty\|_{L^1}\big)C\star S\nonumber\\
&=: -\gamma_t C\star S,
\end{align}
with
\begin{equation}
0\leq \gamma_t\leq 
\begin{cases}
\frac{1}{tm^2 m_t^2}\quad &\text{if }d=2,\\
c\log\Big(1+\frac{1}{m^2t}\Big)\quad &\text{if }d=3,
\end{cases}
\end{equation}
and the bound on $\gamma_t$ obtained in Lemma~\ref{lemm_counterterm}. Thus, for the $\|E\|_{L^\infty}$ norm:
\begin{align}
\Big\|C\star({\bf 1}^\epsilon_0  \|C^3\|_{L^1} - {\bf 1}^\epsilon_0  \|C_\infty^3\|_{L^1})\star S\Big\|_{L^\infty}
&\leq \gamma_t \big(\|C\|_{L^1}\|E\|_{L^\infty} + \|C\|_{L^2}^2\big)\nnb
&= \frac{\gamma_t\|E\|_{L^\infty}}{m^2_t} + \gamma_t\|C\|_{L^2}^2\nonumber\\
&\leq \frac{\gamma_t\|E\|_{L^\infty}}{m^2_t} 
+ \gamma_t\cdot 
\begin{cases}
\frac{c}{m_t^2}\quad \text{if }d=2,\\
\frac{c}{m_t}\quad \text{if }d=3.
\end{cases}
\end{align}
Similarly, the $\|\cdot\|_{L^1}$ norm reads:
\begin{align}
\Big\|C\star({\bf 1}^\epsilon_0  \|C^3\|_{L^1} - {\bf 1}^\epsilon_0  \|C_\infty^3\|_{L^1})\star S\Big\|_{L^1}
\leq
\frac{\gamma_t}{m^2_t}\Big(\|E\|_{L^1}+\frac{1}{m^2_t}\Big),\label{eq_L1_norm_second_lambda2_term} 
\end{align}
so that:
\begin{equation}
\Big\|C\star({\bf 1}^\epsilon_0  \|C^3\|_{L^1} - {\bf 1}^\epsilon_0  \|C_\infty^3\|_{L^1})\star S\Big\|_{L^1\cap L^\infty} 
\leq 
\frac{\gamma_t\|E\|_{L^1\cap L^\infty}}{m^2_t}
 + \gamma_t\Big(
 \frac{1}{m^4_t} + \frac{c}{m^2_t}{\bf 1}_{d=2} + \frac{c}{m_t}{\bf 1}_{d=3}
\Big).
\end{equation}
The last $\lambda^2$ term is $C\star \psi\star S$, with: 
\begin{align}
\psi := (C^3- {\bf 1}^\epsilon_0   \|C^3\|_{L^1}).
\end{align}
One then has:
\begin{align}
\big\|C\star \psi \star S\big\|_{L^1\cap L^\infty}
\leq 
\big\|C\star \psi\big\|_{L^1}\|E\|_{L^1\cap L^\infty} 
+ \big\|C\star \psi\star C\big\|_{L^1\cap L^\infty}.
\end{align}
In Lemma~\ref{lemm_psi_term}, 
the norms of the above quantities are estimated.
Indeed, in dimension $d=2$, one has $\|\psi\|_{L^1}\leq c/m^2_t$, giving
\begin{equation}
 \|C\star \psi\|_{L^1} \leq \frac{c}{m_t^4}, \qquad \|C\star\psi\star C\|_{L^1\cap L^\infty}
  \leq \|\psi\|_{L^1} \|C\star C\|_{L^1 \cap L^\infty}
  \leq \frac{c}{m_t^4} + \frac{c}{m_t^6}.
\end{equation}
In dimension $d=3$, Lemma~\ref{lemm_psi_term} shows that
\begin{align}
\|C\star\psi\|_{L^1}
&\leq 
\frac{c}{m_t^{1/2}} + \frac{c}{m_t^{5/2}},\nonumber\\
\|C\star\psi\star C\|_{L^1} 
&\leq 
\|C\|_{L^1}\|C\star\psi\|_{L^1}
\leq \frac{c}{m_t^{5/2}} + \frac{c}{m_t^{9/2}},\nonumber\\
\|C\star\psi\star C\|_{L^\infty} 
&\leq 
\|C\|_{L^2}\|C\star\psi\|_{L^2}
\leq  \frac{c}{m_t}.
\end{align}
As a result, we find:
\begin{align}
\big\|C\star \psi \star S\big\|_{L^1\cap L^\infty}
&\leq c\|E\|_{L^1\cap L^\infty}\Big[\frac{1}{m^4_t}{\bf 1}_{d=2} + \Big(
\frac{1}{m_t^{1/2}}+  \frac{1}{m_t^{5/2}}\Big){\bf 1}_{d=3}\Big] \nonumber\\ 
&\quad + 
c\Big(\frac{1}{m^4_t}+\frac{1}{m_t^{6}}\Big){\bf 1}_{d=2}
+ c\Big(\frac{1}{m_t}+\frac{1}{m_t^{9/2}}\Big){\bf 1}_{d=3}.
\end{align}
For future reference, note also the following better bound on $\|C\star\psi\star S\|_{L^1}$:
\begin{align}
\big\|C\star \psi \star S\big\|_{L^1}
&\leq 
c\Big(
\frac{1}{m^4_t}{\bf 1}_{d=2} + \big(\frac{1}{m_t^{1/2}}+\frac{1}{m_t^{5/2}}\big){\bf 1}_{d=3}
\Big)\Big(
\|E\|_{L^1}+\frac{1}{m^2_t}
\Big).\label{eq_L1_norm_third_lambda2_term}
\end{align}
\subsubsection{Order 3 terms}\label{sec_L1Linfty_bound_T3}
The only remaining term is the $\lambda^3$ term: 
\begin{equation}
T_{\lambda^3} := 54\lambda^3\, C\star Q\star S ,
\quad Q := S \big(S^2\star S^2\big).
\end{equation}
Notice first:
\begin{align}
\|T_{\lambda^3}\|_{L^1\cap L^\infty}
&\leq 
54\lambda^3 \Big(\|C\star Q\star C\|_{L^\infty} + \|C\|_{L^1}\|Q\|_{L^1}\|E\|_{L^\infty} + \|Q\|_{L^1}\|C\star S\|_{L^1}\Big)\nonumber\\
&\leq 
54\lambda^3 \|Q\|_{L^1}\Big(\|C\|^2_{L^2} + \frac{\|E\|_{L^1\cap L^\infty}}{m^2_t} + \frac{1}{m_t^4}\Big).\label{eq_bound_T3_0}
\end{align}
It is therefore enough to estimate $\|Q\|_{L^1}$. 
Writing $S=C+E$, it reads:
\begin{equation}
\|Q\|_{L^1} 
\leq \Big\|(C+E)\Big[\big(C^2+2EC+E^2\big)\star \big(C^2+2EC+E^2\big)\Big]\Big\|_{L^1}.
\end{equation}
Bounding $E^n$ by $|E|\|E\|^{n-1}_{L^\infty}$ for each $n \geq 1$, 
one finds:
\begin{align}
\|Q\|_{L^1}
&\leq \|C(C^2\star C^2)\|_{L^1}
+ \|E\|_{L^\infty}\Big(\|C^2\|^2_{L^1} + 4 \|C(C\star C^2)\|_{L^1}\Big)\nonumber\\
& \quad+ \|E\|^2_{L^\infty} \Big(6\|C^2\|_{L^1}\|C\|_{L^1} + 4\|C (C\star C)\|_{L^1}\Big)\nonumber\\
& \quad+ \|E\|^2_{L^\infty} \Big(2\|E\|_{L^1}\|C^2\|_{L^1} + 8\|E\|_{L^\infty}\|C\|_{L^1}^2\Big)\nonumber\\
& \quad+ 5\|E\|^3_{L^\infty}\|E\|_{L^1}\|C\|_{L^1}\nonumber\\
& \quad+ \|E\|^3_{L^\infty}\|E\|_{L^1}^2.\label{eq_bound_Q_L1}
\end{align}
The first term on the right-hand side of  \eqref{eq_bound_Q_L1} is estimated in Lemma~\ref{lemm_k=0_Q_term} as:
\begin{equation}
\|C(C^2\star C^2)\|_{L^1}
\leq 
\frac{c}{m^4_t}{\bf 1}_{d=2} + \frac{c}{m_t}{\bf 1}_{d=3}.\label{eq_bounds_Q_L1_1} 
\end{equation}
The inequality $C(x)C(y) \leq \frac12 C(x)^2 +\frac12 C(y)^2$
is enough to also obtain:
\begin{equation}
\begin{split}
&\|C(C\star C^2)\|_{L^1}\leq \|C^2\|^2_{L^1}
\leq 
\frac{c}{m_t^4}{\bf 1}_{d=2} + \frac{c}{m^2_t}{\bf 1}_{d=3},
\\
&\|C(C\star C)\|_{L^1} \leq \|C^2\|_{L^1}\|C\|_{L^1}
\leq 
\frac{c}{m_t^4}{\bf 1}_{d=2}+ \frac{c}{m_t^{3}}{\bf 1}_{d=3}.
\end{split} \label{eq_bounds_Q_L1_2}
\end{equation}
Plugging \eqref{eq_bounds_Q_L1_1}--\eqref{eq_bounds_Q_L1_2} into the bound \eqref{eq_bound_Q_L1} on $\|Q\|_{L^1}$ yields:
\begin{align}
\|Q\|_{L^1}
&\leq 
\frac{c}{m^4_t}{\bf 1}_{d=2} + \frac{c}{m_t}{\bf 1}_{d=3}
	+ \|E\|_{L^1\cap L^\infty}\Big(\frac{c}{m_t^4}{\bf 1}_{d=2} + \frac{c}{m^2_t}{\bf 1}_{d=3}\Big)\nonumber\\
& \quad+ 
\|E\|^2_{L^1\cap L^\infty} \Big(\frac{c}{m_t^4}{\bf 1}_{d=2}+ \frac{c}{m_t^{3}}{\bf 1}_{d=3}\Big)\nonumber\\
& \quad+ 
\|E\|^3_{L^1\cap L^\infty} 
\Big(
\frac{c}{m^2_t}{\bf 1}_{d=2} + \frac{c}{m_t}{\bf 1}_{d=3}
	+ \frac{8}{m^4_t}
	\Big)\nonumber\\
& \quad+ \frac{5\|E\|^4_{L^1\cap L^\infty}}{m^2_t} 
	+ \|E\|^5_{L^1\cap L^\infty}.\label{eq_bound_Q_L1_final}
\end{align}
Injecting the bound \eqref{eq_bound_Q_L1_final} into the bound \eqref{eq_bound_T3_0} for $T_{\lambda^3}$, 
one finds that there is a polynomial $P_{3}$ of degree 6, with no constant term, such that:
\begin{align}
\|T_{\lambda^3}\|_{L^1\cap L^\infty}
&\leq
c\lambda^3\Big[\Big(\frac{1}{m^6_t} + \frac{1}{m^8_t}\Big){\bf 1}_{d=2} + \Big(\frac{1}{m^2_t} + \frac{1}{m^5_t}\Big){\bf 1}_{d=3}
	+ P_{3}\big(\|E\|_{L^1\cap L^\infty}\big)
\Big].
\end{align}
Moreover, $P_3$ has coefficients given by sums of negative powers of $m_t$, 
the lowest power in absolute value being $1$ (if $d=3$) or $2$ (if $d=2$). 
\subsubsection{Conclusion on $\|E\|_{L^1\cap L^\infty}$}
Putting together the estimates of $T_{\lambda^i}$ ($0\leq i\leq 3$) obtained in the previous sections, 
one finds for $\|E\|_{L^1\cap L^\infty}$:
\begin{align}
\|E\|_{L^1\cap L^\infty} 
&\leq 
c\big(
\lambda\eta_t + \lambda^2\gamma_t + |\mu-m^2|
\big)
\Big(
\frac{1}{m^2_t}{\bf 1}_{d=2} + \frac{1}{m_t}{\bf 1}_{d=3} + \frac{1}{m^4_t}
\Big)
\nonumber\\
&\quad + 
\frac{c\|E\|_{L^1\cap L^\infty}}{m^2_t}\big(
\lambda\eta_t + \lambda^2\gamma_t + |\mu-m^2|
\big)
+ c\sum_{i=1}^3\lambda^i R^{(i)}_{t}\big(\|E\|_{L^1\cap L^\infty}\big),\label{eq_bound_L_infty_norm_interm} 
\end{align}
where, for $i\in\{1,2,3\}$, $R^{(i)}_{t}$ is a polynomial with positive coefficients given by sums of negative powers of $m_t$ and with no constant term. 
The coefficient of the $R^{(i)}_t$ with the lowest power of $m_t$ in absolute value is proportional to $m_t^{-1/2}$ in dimension $d=3$, and to $m_t^{-2}$ in dimension $d=2$.

The bound \eqref{eq_bound_L_infty_norm_interm} yields the following estimates of $\|E\|_{L^1\cap L^\infty}$.

\begin{lemma}\label{lemm_bound_L1_Linfty_norm}
Let $d\in\{2,3\}$, $\lambda>0$, $\mu\in\R$, and $m^2>0$. 
\begin{itemize}
	\item[(i)] (Small scale). There is $t_0(d,\lambda,\mu,m^2)>0$ (independent of $\epsilon,L$) and a numerical constant $c>0$ such that:
        \begin{equation}
          \forall t\in (0,t_0(d,\lambda,\mu,m^2)],
          \qquad\|E\|_{L^1\cap L^\infty}\leq c\lambda.
        \end{equation}
	\item[(ii)] (All scales, one small coupling). If $\mu>0$ and with the special choice $m^2=\mu$, 
	there are parameters $\lambda(d,\mu),\mu(d,\lambda)>0$ (independent of $\epsilon,L$) with the following property.
	If either $\lambda\in[0,\lambda(d,\mu)]$ or $\mu\geq \mu(d,\lambda)$, 
	then there is $c(\mu)>0$ with: 
	\begin{equation}
          \forall t> 0,\qquad
          \|E\|_{L^1\cap L^\infty}\leq c(\mu)\lambda,
          \qquad 
          \lim_{\mu\rightarrow\infty}c(\mu)=0.
\end{equation}
\end{itemize} 
\end{lemma}

In the proof of the lemma, we will use the following qualitative properties of $t\mapsto E=E_t$,
which hold for any fixed $\epsilon>0$ and $L<\infty$:
\begin{equation}
  \begin{split}
    &t\in (0,\infty) \mapsto \|E\|_{L^1\cap L^\infty} \text{ is continuous};
  \\
  &\lim_{t\to 0} \|E\|_{L^1\cap L^\infty} =0.\label{eq_properties_E}
\end{split}
\end{equation}
Indeed, since $\epsilon>0$ and $L<\infty$ are fixed, the above continuity follows from the
continuity of $t\mapsto S(x)$ and $t\mapsto C(x)$ for each $x\in\Lambda_{\epsilon,L}$.
Similarly, the $t\to 0$ limit follows from
$\lim_{t\to 0} C(x) = \lim_{t\to 0}S(x)=0$ for each $x\in\Lambda_{\epsilon,L}$.
This is clear for $C(x) = (-\Delta^\epsilon+m^2+1/t)^{-1}(0,x)$,
and for $S(x)$ it follows, for example, from the
Brascamp--Lieb inequality which shows 
$S(x) \leq S(0) \leq
(-\Delta^\epsilon+a^\epsilon(\lambda,m^2)+\mu+1/t)^{-1}(0,0)$
provided $a^\epsilon(\lambda,m^2)+\mu+1/t > 0$ (similarly to \eqref{e:BLappl}).

\begin{proof}
Let $d\in\{2,3\}$, and let $f_{\lambda,\mu,m^2,t}$ be the polynomial
such that \eqref{eq_bound_L_infty_norm_interm} corresponds to:
\begin{equation}
\|E\|_{L^1\cap L^\infty}
\leq \frac{c\lambda\eta_t}{m_t}{\bf 1}_{d=3} + f_{\lambda,\mu,m^2,t}\big(\|E\|_{L^1\cap L^\infty}).\label{eq_def_f_lambda}
\end{equation}
The reason why the first term is taken out is due to \eqref{eq_limit_counterterms_small_t} below.\\

Consider first item (i). The difference of the counterterms $\eta_t,\gamma_t$ are estimated in Lemma~\ref{lemm_counterterm}:
\begin{equation}  \label{eq_counterterm_bis}
\begin{split}
&\eta_t \leq c\log\Big(1+\frac{1}{m^2t}\Big){\bf 1}_{d=2} + c m\Big(\sqrt{1 + \frac{1}{tm^2}} - 1\Big){\bf 1}_{d=3},\\
&
\gamma_t 
\leq 
\frac{c}{m^2(m^2t+1)} {\bf 1}_{d=2} 
+ 
c\log\Big(1+\frac{1}{m^2t}\Big){\bf 1}_{d=3}.
\end{split}
\end{equation}
In particular, they satisfy, for a numerical constant $c_0>0$:
\begin{align}
\sup_{t>0}\frac{c\eta_t}{m_t}\leq {\bf 1}_{d=2} c(m^2) + c_0{\bf 1}_{d=3} \quad \text{with}\quad \lim_{m\rightarrow\infty} c(m^2)=0,\label{eq_bound_counterterms}
\end{align}
and:
\begin{equation}
\limsup_{t\downarrow0}\frac{c\eta_t}{m_t}\leq 0{\bf 1}_{d=2}+ c_0{\bf 1}_{d=3},\qquad
\lim_{t\downarrow0}\frac{\gamma_t}{m_t}=0.\label{eq_limit_counterterms_small_t}
\end{equation}
On the other hand, note the following elementary property:
\begin{equation}
m^2_t := m^2 + \frac{1}{t} 
\quad\Rightarrow\quad 
\lim_{t\downarrow 0}m_t = +\infty.
\end{equation}

It follows from the above that the (positive) coefficients of the polynomial $f_{\lambda,\mu,m^2,t}$ tend to $0$ as $t\downarrow 0$.
In particular, there is a largest value $t_0(d,\lambda,\mu,m^2)>0$ such that:
\begin{equation}
\forall t\in (0,t_0(d,\lambda,\mu,m^2)],
\qquad
f_{\lambda,\mu,m^2,t}(2c_0\lambda)\leq \frac{c_0\lambda}{2},\label{eq_bound_f_lambda_small_t}
\end{equation}
where $c_0$ is the constant from \eqref{eq_bound_counterterms} (and we assume $c(m^2) \leq c_0$ in $d=2$).
Note also that $x\mapsto f_{\lambda,\mu,m^2,t}(x)$ is increasing in $x\geq 0$
(the coefficients of $f_{\lambda,\mu,m^2,t}$ are positive).
Since $t\in (0,\infty)\mapsto \|E_t\|_{L^1\cap L^\infty}$ is continuous
and $\lim_{t\to\infty} \|E_t\|_{L^1\cap L^\infty} =0$, both by \eqref{eq_properties_E}, 
it follows that $\|E\|_{L^1\cap L^\infty}\leq 2c_0\lambda$ on $(0,t_0(d,\lambda,\mu,m^2)]$. 
Indeed, suppose the contrary, 
and take $t_*\in (0,t_0(d,\lambda,\mu,m^2)]$ with $\|E_{t_*}\|_{L^1\cap L^\infty}> 2c_0\lambda$. 
Then \eqref{eq_bound_f_lambda_small_t} and the definition \eqref{eq_def_f_lambda} of $f_{\lambda,\mu,m^2,t}$ yield a contradiction:
\begin{equation}
2c_0\lambda
< \|E_{t_*}\|_{L^1\cap L^\infty}
\leq 
c_0\lambda + f_{\lambda,\mu,m^2,t_*}(2c_0\lambda)
\leq 
\frac{3c_0\lambda}{2}
< 
2c_0\lambda.
\end{equation}
This concludes the proof of item (i).\\

For item (ii), the reasoning is similar, so we only give a sketch.
As $\mu=m^2$, the right-hand side of \eqref{eq_bound_L_infty_norm_interm} only contains linear combinations of the following terms:
\begin{gather}
  c(\lambda\eta_t+\lambda^2\gamma_t)\Big(\frac{1}{m_t^2}{\bf 1}_{d=2} + \frac{1}{m_t}{\bf 1}_{d=3}\Big),
  \qquad
  c(\lambda\eta_t+\lambda^2\gamma_t)m_t^{-a}, \qquad a\geq 2,
  \nnb
  \text{and} \quad 
c\lambda^n m_t^{-a},\qquad a>0,n\in\{1,2,3\}.
\end{gather}
As for item (i), none of these terms depend on $\epsilon$ or $L$, and $c\lambda^nm_t^{-a}$ is bounded by $c\lambda^n m^{-a} = c\lambda^n \mu^{-a/2}$ for each $a,n$.  
To see what happens for large $\mu$ or small $\lambda$, 
recall from \eqref{eq_bound_counterterms} the bound:
\begin{equation}
\sup_{t>0}\frac{c\eta_t}{m_t}
\leq c(\mu){\bf 1}_{d=2} +c_0
\quad\text{with} \quad \lim_{\mu\rightarrow\infty}c(\mu)=0,
\end{equation}
and from \eqref{eq_counterterm_bis} the corresponding bound for $\gamma_t$:
\begin{equation}
\sup_{t>0}\frac{c\gamma_t}{m_t}\leq c'(\mu) \quad \text{with} \quad \lim_{\mu\rightarrow\infty}c'(\mu)=0.
\end{equation}
It follows that the right-hand side of \eqref{eq_bound_L_infty_norm_interm} is bounded uniformly in $t$ by $c(\mu)\lambda$ 
(for a different $c(\mu)$ that nonetheless vanishes when $\mu$ is large),  
provided either $\lambda$ is smaller than some $\lambda(d,\mu)>0$ or $\mu$ larger than some $\mu(d,\lambda)>0$.
Thus one has an analogue of \eqref{eq_bound_f_lambda_small_t} with $t_0=+\infty$ and
repeating the argument of item (i) concludes the proof of item (ii) and of the lemma.
\end{proof}

\subsection{Estimate of the susceptibility}
Fix a dimension $d\in\{2,3\}$, 
let $\lambda>0$, $\mu\in\R$, $m^2>0$, and $t> 0$. 
We now obtain the bound of Proposition \ref{prop_bound_susceptibility} on the norm $\|E\|_{L^1}$.   
The starting point is again the formula \eqref{eq_bound_BFS_3d} with the splitting \eqref{eq_splitting_counterterms}:
\begin{align}
\big|S(x)-C(x)\big| 
&\leq 
3\lambda\big(\eta_t + \|E\|_{L^\infty}\big)(C\star S)(x) + 6\lambda^2\big|\big(C \star \big(S^3 - {\bf 1}^\epsilon_0  \|C_\infty^3\|_{L^1}\big) \star S\big)(x)\big| \nnb
&\quad 
+ 54\lambda^3 (C\star Q\star S)(x) + |\mu-m^2|\big(C\star S\big)(x).\label{eq_starting_point_susceptibility_3d}
\end{align}
The quantity $\|E\|_{L^\infty}$ is estimated by Lemma~\ref{lemm_bound_L1_Linfty_norm}. 
The $\|\cdot\|_{L^1}$ norm of $E$ then reads:
\begin{align}
\|E\|_{L^1}\leq \frac{3\lambda\big(\eta_t + \|E\|_{L^\infty} + |\mu-m^2|\big)}{m^2_t}\Big(\|E\|_{L^1} + \frac{1}{m^2_t}\Big)
+ \|T_{\lambda^2}\|_{L^1}
+ \|T_{\lambda^3}\|_{L^1}.
\end{align}
\begin{lemma}\label{lemm_first_bound_susceptibility}
There are polynomials $P^{(i)}_{t} = P^{(i)}_{d,t,\mu,m^2,\|E\|_{L^\infty}}$ with $i\in\{1,2,3\}$, 
of degree at most $3$ and with no constant term, 
such that:
\begin{align}
\|E\|_{L^1}
&\leq
\frac{3\lambda\big(\eta_t+\|E\|_{L^\infty} + |\mu-m^2|\big)}{m^4_t} \nnb
&\quad
+ \frac{c\lambda^2}{m^4_t}\Big[
\gamma_t + \frac{\|E\|_{L^\infty} + \|E\|^2_{L^\infty}}{m^2_t} 
+ \frac{1}{m^2_t}{\bf 1}_{d=2} 
+ \Big(\frac{1}{m_t} + \frac{1}{m_t^{1/2}}\Big){\bf 1}_{d=3}
\Big]\nonumber\\
&\quad+
\lambda^3\Big[
\frac{c}{m_t^8}
	 + \Big(
	\frac{c}{m_t^{7}}
	+ \frac{c}{m^5_t}
	\Big){\bf 1}_{d=3}
	\Big]
+ \sum_{i=1}^3\lambda^iP^{(i)}_{t}\big(\|E\|_{L^1}\big).
\end{align}
The coefficients of the polynomials $P^{(i)}_t$ are nonnegative and of the following form:
\begin{equation}
c\frac{\eta_t+ |\mu-m^2|}{m^4_t}\quad \text{or}\quad
c\frac{\gamma_t}{m^2_t}\quad \text{or}\quad
c\|E\|_{L^\infty}^nm_t^{-a},\qquad 
c,a>0;\ n\in\{0,1,2,3\}.
\end{equation}
The lowest value of $a$ is $2$ in dimension 2, and $1/2$ in dimension 3.
\end{lemma}
\begin{proof}
The norms $\|T_{\lambda^2}\|_{L^1}$ and $\|T_{\lambda^3}\|_{L^1}$ have been estimated in Sections \ref{sec_L1Linfty_bound_T2}--\ref{sec_L1Linfty_bound_T3} respectively. 
Precisely, $\|T_{\lambda^2}\|_{L^1}$ is estimated in \eqref{eq_L1_norm_first_lambda2_term}--\eqref{eq_L1_norm_second_lambda2_term}--\eqref{eq_L1_norm_third_lambda2_term} and, 
bounding $E^n$ by $|E| \|E\|_{L^\infty}^{n-1}$, it reads:
\begin{align}
\lambda^{-2}\|T_{\lambda^2}\|_{L^1}
&\leq 
\frac{c\gamma_t}{m^4_t} 
+ \frac{c(\|E\|_{L^\infty}+1)}{m^6_t}{\bf 1}_{d=2} 
+  c\Big(
\frac{\|E\|_{L^\infty}}{m^5_t} + \frac{1}{m_t^{5/2}} + \frac{1}{m_t^{9/2}}
\Big){\bf 1}_{d=3}
+ \frac{\|E\|_{L^\infty}^2}{m^6_t}
\nonumber\\
&\quad +
c\|E\|_{L^1}\Big[
\frac{\gamma_t}{m^2_t} + \frac{\|E\|_{L^\infty}^2}{m^4_t}
+ \frac{(\|E\|_{L^\infty}+1)}{m^4_t}{\bf 1}_{d=2}
+ \Big(
\frac{1}{m^{1/2}_t} + \frac{1}{m_t^{5/2}} + \frac{\|E\|_{L^\infty}}{m^3_t}
\Big){\bf 1}_{d=3}
\Big] 
\nnb
&\quad+
\frac{c\|E\|^2_{L^1}}{m^2_t}.
\end{align}
On the other hand, 
for $\|T_{\lambda^3}\|_{L^1} = 54\lambda^3 \|C\star S\|_{L^1}\|Q\|_1$, one has:
\begin{equation}
\lambda^{-3}\|T_{\lambda^3}\|_{L^1}
\leq 
\frac{54}{m^2_t}\Big(
\|E\|_{L^1} + \frac{1}{m^2_t}
\Big)\|Q\|_{L^1}.
\end{equation}
Recalling the bound \eqref{eq_bound_Q_L1} on $\|Q\|_{L^1}$ concludes the proof of Lemma \ref{lemm_first_bound_susceptibility}:
\begin{align}
\|Q\|_{L^1}
&\leq 
\frac{\|E\|_{L^\infty}^3}{m^4_t}+ \frac{c\big(1+\|E\|_{L^\infty}+\|E\|_{L^\infty}^2\big)}{m_t^4}{\bf 1}_{d=2}
+ c\Big(
\frac{1}{m_t} + \frac{\|E\|_{L^\infty}}{m_t^2} + \frac{\|E\|_{L^\infty}^2}{m^3_t}
\Big){\bf 1}_{d=3}\nonumber\\
& \quad+ 
c\|E\|_{L^1}\Big[
	\frac{\|E\|_{L^\infty}^3}{m^2_t} 
	+ \|E\|_{L^\infty}^2\Big(
	\frac{1}{m^2_t}{\bf 1}_{d=2} + \frac{c}{m_t}{\bf 1}_{d=3}
	\Big)
\Big]
	+ \|E\|_{L^\infty}^3\|E\|_{L^1}^2.
\end{align}
\end{proof}
Lemma \ref{lemm_first_bound_susceptibility} is sufficient to conclude the proof of Propositions \ref{prop_bound_susceptibility}--\ref{prop_bound_susceptibility_long_time}.

\begin{proof}[Proof of Proposition \ref{prop_bound_susceptibility}]
Let $t_0(d,\lambda,\mu,m^2)>0$ be the scale defined in Lemma~\ref{lemm_bound_L1_Linfty_norm}, 
such that, for some numerical constant $c$:
\begin{equation}
\forall t\in(0,t_0(d,\lambda,\mu,m^2)],\qquad
\|E\|_{L^1\cap L^\infty}\leq c\lambda.\label{eq_bound_E_L1_Linfty}
\end{equation}
Henceforth $t\in (0,t_0(d,\lambda,\mu,m^2)]$. 
Let $G_{d,t,\mu,m^2}(\lambda)+\sum_{i=1}^3\lambda^iP^{(i)}_{d,t,\mu,m^2,c\lambda}\big(\|E\|_{L^1}\big)$ denote the right-hand side of the bound in Lemma~\ref{lemm_first_bound_susceptibility} with $\|E\|_{L^\infty}$ replaced by $c\lambda$:
\begin{equation}
\|E\|_{L^1}\leq G_{d,t,\mu,m^2}(\lambda)+\sum_{i=1}^3\lambda^i P^{(i)}_{d,t,\mu,m^2,c\lambda}\big(\|E\|_{L^1}\big).\label{eq_bound_E_L1_interm}
\end{equation}

Using the already established bound $\|E\|_{L^1}\leq c\lambda$ on the right-hand side
would give a bound on $\|E\|_{L^1}$ in which each term
has a factor $m_t^{-a} \sim t^{a/2}$ as $t\to 0$.
However, the smallest powers obtained in this way (for example, $a=1/2$ in dimension $3$)
are not integrable with respect to $t^{-2}\, dt$ near $t=0$.
On the other hand,
the term in $G_{d,t,\mu,m^2}(\lambda)$ with the most singular $t$-dependence as $t\to 0$
is $\eta_tm_t^{-4}$ which is integrable with respect to $t^{-2} \, dt$.
(Indeed, recall
$m_t^{-4} \sim t^2$, that \eqref{eq_counterterm_bis} shows that $\eta_t$ diverges as $\log t^{-1}$ in $d=2$
and as $t^{-1/2}$ in $d=3$, and that $\gamma_t$ is bounded in $d=2$ whereas $\gamma_t$ diverges as $\log t^{-1}$ in $d=3$). In addition:
\begin{equation}
\sup_{t>0}G_{d,t,\mu,m^2}(\lambda)\leq c(d,\lambda,\mu,m^2).
\end{equation}
Therefore, since the $P^{(i)}$ do not have constant terms,
by iterating the bound \eqref{eq_bound_E_L1_interm} a finite number of times, 
and using the bound $\|E\|_{L^1}\leq c\lambda$ from \eqref{eq_bound_E_L1_Linfty} at the last iteration, it follows that
\begin{equation}
\|E\|_{L^1} \leq p_{d,t,\mu,m^2}(\lambda),\label{eq_def_polynomial_p_t_lambda}
\end{equation}
with the right-hand side above a polynomial in $\lambda$
of degree independent of the parameters $\lambda,\mu,m^2,\epsilon,L$,
and coefficients functions of $\eta_t,\gamma_t, |\mu-m^2|$ times powers of $m_t^{-1}$ such that
\begin{equation}
t\mapsto \frac{p_{d,t,\mu,m^2}(\lambda)}{t^2}
\qquad \text{is integrable around $t=0$.}
\end{equation}
This concludes the proof of Proposition \ref{prop_bound_susceptibility}.
\end{proof}

\begin{proof}[Proof of Proposition \ref{prop_bound_susceptibility_long_time}]
Fix $\mu=m^2$, 
and consider either $\lambda\in[0,\lambda(d,\mu)]$, 
or $\mu\geq \mu(d,\lambda)$, 
where these quantities are defined in item (ii) of Lemma \ref{lemm_bound_L1_Linfty_norm}. 
Then:
\begin{equation}
\forall t>0,\qquad
\|E\|_{L^1\cap L^\infty}\leq c(\mu)\lambda,\qquad
\lim_{\mu\rightarrow\infty}c(\mu) =0.
\end{equation}
The argument is then identical to the above proof of Proposition \ref{prop_bound_susceptibility}, 
replacing $\|E\|_{L^\infty}$ by $c(\mu)\lambda$ for each $t>0$ and writing the right-hand side of the bound in Lemma~\ref{lemm_first_bound_susceptibility} as:
\begin{equation}
\|E\|_{L^1} \leq \tilde G_{d,t,\mu}(\lambda)+\sum_{i=1}^3\lambda^i P^{(i)}_{d,t,\mu,\mu,c(\mu)\lambda}\big(\|E\|_{L^1}\big).
\end{equation}
Then, again by iteration, for each $t>0$:
\begin{equation}
\|E\|_{L^1}\leq \tilde p_{d,t,\mu}(\lambda),
\end{equation}
with $t\mapsto t^{-2}\tilde p_{d,t,\mu}(\lambda)$ integrable around $0$. 
Using the bound of Lemma \ref{lemm_first_bound_susceptibility} together with the estimates on the differences of counterterms from Lemma~\ref{lemm_counterterm}, 
it also follows that $\tilde p_{d,t,\mu}(\lambda)$ has coefficients bounded uniformly on $t$.
It is therefore integrable at $\infty$ with respect to $t^{-2}\, dt$, and
thus on the whole of $(0,\infty)$. 
Moreover, since each term appearing in $\tilde p_{d,t,\mu}(\lambda)$ vanishes when $\lambda\to 0$ or $\mu\to\infty$,
one has, 
up to decreasing $\lambda(d,\mu)$ or increasing $\mu(d,\lambda)$:
\begin{equation}
\sup_{\lambda>0}\sup_{\mu\geq \mu(d,\lambda)}\int_0^\infty\frac{\tilde p_{d,t,\mu}(\lambda)}{t^2}\, dt <\infty,
\qquad 
\sup_{\mu>0}\sup_{\lambda\in [0,\lambda(d,\mu)]}\int_0^\infty\frac{\tilde p_{d,t,\mu}(\lambda)}{t^2}\, dt <\infty.
\end{equation}
The final claim $\limsup_{t\to\infty} \tilde p_{d,t,\mu}(\lambda) = o(\mu^{-1})$ is also clear from the above 
(recall that the differences $\eta_t,\gamma_t$ of the counterterms at time $t$ and $\infty$ vanish by definition in the limit $t\to\infty$,
so only the inverse powers of $m_t$ in $\tilde p_{d,t,\mu}$ matter in the limit).
\end{proof}

\section{Proof of Theorem~\ref{thm:lsiphi4d}}
\label{sec:proofs}

\begin{proof}[Proof of Theorem~\ref{thm:lsiphi4d} (i)]
Fix $m^2=1$ (though the value of $m^2$ is irrelevant).
Let $\bar\chi\in (0,+\infty)$ and assume that $\epsilon$ and $L$ are such that
\begin{equation}
\chi^{\epsilon,L}_{\infty} := \epsilon^d\sum_{x\in\Lambda_{\epsilon,L}}\big<\varphi_0\varphi_x\big>^{\epsilon,L}_{\lambda,\mu} \leq  \bar\chi.\label{eq_hyp_susceptibility_bounded}
\end{equation}
By Theorem~\ref{thm:criterion_continuum}, the log-Sobolev constant $\gamma = \gamma^{\epsilon,L}(\lambda,\mu)$ of the continuum $\varphi^4_d$ measure \eqref{eq_def_continuum_phi4_measure} is bounded by
\begin{equation}
  \frac{1}{\gamma}\leq \int_0^\infty\exp\Big[-2\int_0^t\dot\kappa^{\epsilon,L}_s \, ds\Big]\, dt,
  \qquad
    \dot\kappa^{\epsilon,L}_t := \frac{1}{t} - \frac{\chi_t^{\epsilon,L}}{t^2}.
\end{equation}
To bound $\gamma$, we estimate the susceptibility for small scale $t$ and large scale $t$ in a different way.

For $t$ smaller than the $t_0:= t_0(d,\lambda,\mu,m^2=1)>0$ of Proposition \ref{prop_bound_susceptibility}, one has:
\begin{equation}
\chi^{\epsilon,L}_t \leq \frac{1}{1+\frac{1}{t}} + p_{d,t,\mu,m^2=1}(\lambda),
\qquad
\int_{0}^{t_0}\frac{p_{d,t,\mu,m^2=1}(\lambda)}{t^2}\, dt \leq c_d(\lambda,\mu,m^2=1),
\end{equation}
where the $(1+1/t)^{-1}$ term is simply the Gaussian susceptibility $\|C_t\|_{L^1}$ with mass $m^2=1$. 
This gives:
\begin{equation}
\forall t\in(0,t_0],\qquad  
\dot\kappa^{\epsilon,L}_t
  \geq  \frac{1}{t}-\frac{1}{t^2}\pa{\frac{1}{1+1/t} + p_{d,t,\mu,m^2=1}(\lambda)}
  = \frac{1}{t+1} - \frac{p_{d,t,\mu,m^2=1}(\lambda)}{t^2}.
\end{equation}
and
\begin{equation}
  \int_0^t \dot \kappa^{\epsilon,L}_s \, ds \geq \log(t+1) - c_d(\lambda,\mu,m^2=1).\label{eq_bound_kappa_small_time}
\end{equation}

On the other hand, the second Griffiths inequality implies that $\chi_t^{\epsilon,L}$ is decreasing in $1/t$, i.e., increasing in $t$. 
For large scales $t>t_0$, one thus has $\chi_t^{\epsilon,L} \leq \chi^{\epsilon,L}_\infty \leq \bar\chi$ by \eqref{eq_hyp_susceptibility_bounded}, so that:
\begin{equation}
\forall t>t_0,\qquad  
\dot\kappa_t^{\epsilon,L} \geq \frac{1}{t} - \frac{\bar\chi}{t^2}.
\end{equation}
This bound implies:
\begin{equation}
 \forall t>t_0,\qquad 
 \int_{t_0}^t \dot\kappa^{\epsilon,L}_s \, ds \geq \log t - \log t_0 - \bar\chi/t_0 = \log (t/t_0) - C(\lambda,\mu,\bar\chi).\label{eq_bound_kappa_large_time}
\end{equation}

Combined, \eqref{eq_bound_kappa_small_time} and \eqref{eq_bound_kappa_large_time} give for each $t>0$:
\begin{align}
  \exp\qa{-2\int_0^t \kappa_s^{\epsilon,L}\, ds}
  &\leq
  \exp\qa{-2\int_0^{t_0\wedge t} \kappa_s^{\epsilon,L}\, ds}
    \exp\qa{-2\int_{t_0\wedge t}^t \kappa_s^{\epsilon,L}\, ds}
    \nnb
  &\leq (t\wedge t_0+1)^{-2} (1\vee t/t_0)^{-2} e^{2c_d(\lambda,\mu,m^2=1)+2C(\lambda,\mu,\bar\chi)}
\end{align}
which is integrable on $(0,\infty)$ with a bound that depends only on $(\lambda,\mu,\bar\chi)$.
\end{proof}

\begin{proof}[Proof of Theorem~\ref{thm:lsiphi4d} (ii)]
 Under the assumptions of (ii) ($m^2>0$, $\mu=m^2$ and either $\lambda\in[0,\lambda(d,\mu)]$ or $\mu\geq \mu(d,\lambda)$), we can directly apply Proposition~\ref{prop_bound_susceptibility_long_time} for all $t>0$.
  As in the proof of (i), this gives
  \begin{equation}
    \dot\kappa_t^{\epsilon,L} \geq  \frac{m^2}{m^2t+1}-\frac{\tilde p_{d,t,m^2}(\lambda)}{t^2}
  \end{equation}
  and
  \begin{equation}
    \kappa_t^{\epsilon,L} \geq \log(m^2t+1) - c,\qquad c:= \int_0^\infty\frac{\tilde p_{d,t,m^2}(\lambda)}{t^2}.
  \end{equation}
  Recall from Proposition~\ref{prop_bound_susceptibility_long_time} that $c$ is independent of $\lambda,\mu$ under the assumptions of (ii).
  The log-Sobolev constant $\gamma = \gamma^{\epsilon,L}(\lambda,\mu)$ thus satisfies:
  \begin{equation}
   \frac{1}{\gamma}
  \leq
  \int_0^\infty \exp\qa{-2\int_0^t \kappa_s^{\epsilon,L}\, ds}\ dt
  \leq
  e^{c}\int_0^\infty \frac{1}{(m^2t +1)^2}\, dt
  = \frac{e^{c}}{m^2}.
\end{equation}
  Thus $\gamma\geq e^{-c}m^2 = e^{-c}\mu$.
  Since $\chi^{\epsilon,L}(\lambda,\mu)= \|S_\infty\|_{L^1}$ and $\|S_\infty\|_{L^1} = \mu^{-1} + o(\mu^{-1})$
  by the last item of Proposition~\ref{prop_bound_susceptibility_long_time},
  the bound on $\gamma$ in the other direction was established just below the statement of Theorem~\ref{thm:lsiphi4d}.
  This concludes the proof of item (ii) of Theorem~\ref{thm:lsiphi4d}.
\end{proof}

\appendix
\section{Bounds on diagrams}\label{app_bounds_on_diagram}
\subsection{Moments of Gaussian covariance}
In this section, $m^2>0$ is fixed and we consider:
\begin{equation}
C_{(m)} := \big(-\Delta^{\epsilon} + m^2\big)^{-1},\label{eq_Gaussian_covariance_appendix}
\end{equation}
with the matrix elements of the inverse normalised with respect to the inner product $(u,v)_\epsilon = \epsilon^d \sum_{x\in\Lambda_{\epsilon,L}} u(x)v(x)$. 
In other words, $C_{(m)}$ is given by:
\begin{equation}
\forall x\in\Lambda_{\epsilon,L},\qquad 
C_{(m)}(x) 
:= 
C_{(m)}(0,x)
= 
\frac{1}{L^d}\sum_{k\in\Lambda^*}e^{i k\cdot x}\frac{1}{m^2+\theta(k)}.
\end{equation}
Above, $\Lambda^*$ is the set:
\begin{equation}
  \Lambda^* := \Big\{k\in\frac{2\pi}{ L}\Z^d : \forall i\in\{1,\dots,d\},\ -\frac{\pi}{\epsilon}<k_i\leq  \frac{\pi}{\epsilon}\Big\}\label{eq_def_Lambda*}
  ,
\end{equation}
and $\Delta^\epsilon$ is the Laplacian on the torus $\Lambda_{\epsilon,L}$ defined in \eqref{eq_def_continuous_Laplacian}, 
with eigenvalues $\theta$ given by:
\begin{equation}
\theta(k) 
:= \sum_{k=1}^3 \theta(k_i),
\qquad
\theta(k_i) 
:= \frac{4}{\epsilon^2}\sin^2\Big(\frac{k_i\epsilon}{2}\Big),
\qquad k\in\Lambda^*.\label{eq_def_theta}
\end{equation}
For later use, recall that convexity of the sine function on $[0,\pi/2]$ yields, for each $0\leq u \leq \pi/2$:
\begin{equation}
\frac{2u}{\pi}
\leq 
\sin(u)
\leq u
\quad \Rightarrow\quad 
\frac{4u^2}{\pi^2}
\leq 
\sin^2(\pm u)
\leq u^2.
\end{equation}
In particular:
\begin{equation}
\forall k\in\Lambda^*,
\qquad
\frac{1}{m^2+\theta(k)}
\leq 
\frac{1}{m^2+\frac{4\|k\|_2^2}{\pi^2}}.\label{eq_bound_theta}
\end{equation}
\begin{lemma}\label{lemm_psi_term}
Let $\psi$ be given by:
\begin{equation}
\psi = C_{(m)}^3 - {\bf 1}^\epsilon_0\|C_{(m)}^3\|_{L^1},
\end{equation}
where $C_{(m)}$ is the Gaussian covariance \eqref{eq_Gaussian_covariance_appendix}, and ${\bf 1}^\epsilon_0 := \epsilon^{-d}{\bf 1}_0$. 
In dimension $d=2$, one has: 
\begin{equation}
  \|\psi\|_{L^1}\leq \frac{c}{m^2}.
\end{equation}
When instead $d=3$, for each $k\in\Lambda^*$, 
the Fourier transform $\hat\psi(k)$ of $\psi$ satisfies:
\begin{equation}
  |\hat\psi(k)|\leq c\log\Big(1+\frac{2\|k\|_2}{m\pi}\Big),
\end{equation}
and
\begin{equation}
  \|C_{(m)}\star\psi\|_{L^1}
  \leq \frac{c}{m^{1/2}} + \frac{c}{m^{5/2}},\qquad
  \|C_{(m)}\star\psi\|_{L^2}\leq \frac{c}{m^{1/2}}.
\label{eq_bound_psi}
\end{equation}
\end{lemma}

\begin{proof}
Let us start with the bound on $\|\psi\|_{L^1}$ in dimension $d=2$, which is easiest. 
Indeed, one has in that case:
\begin{equation}
|\psi| 
\leq 
C^3_{(m)} + {\bf 1}^\epsilon_0\|C^3_{(m)}\|_{L^1}
\quad \Rightarrow\quad 
\|\psi\|_{L^1}\leq 2\|C^3_{(m)}\|_{L^1}.
\end{equation}
One has then:
\begin{align}
\|C^3_{(m)}\|_{L^1} = \frac{1}{L^{2d}}\sum_{k,p\in\Lambda^*}\frac{1}{m^2 + \theta(p)}
\cdot \frac{1}{m^2 + \theta(q)}
\cdot\frac{1}{m^2+\theta(p+q)} .
\end{align}
Using convexity to bound the sines in the definition \eqref{eq_def_theta} of $\theta$ and comparing sums to integrals yields:
\begin{align}
\|C^3\|_{L^1}
&\leq 
\int_{[-\frac{\pi}{\epsilon}-1,\frac{\pi}{\epsilon}]^4}dp\,dq\,  \frac{1}{m^2 + \frac{4\|p\|_2^2}{\pi^2}}
\cdot \frac{1}{m^2 + \frac{4\|q\|_2^2}{\pi^2}}
\cdot \frac{1}{m^2 + \frac{4\|p + q\|_2^2}{\pi^2}}\nnb
&\leq 
\frac{\pi^4}{16m^2}\int_{\R^4}dp\,dq\,  \frac{1}{1 + \|p\|_2^2}
\cdot \frac{1}{1 + \|q\|_2^2}
\cdot \frac{1}{1 + \|p + q\|_2^2}
\leq \frac{c}{m^2},
\label{eq_bound_psi_L1_d2}
\end{align}
where we used that the integral on the last line is finite which
follows, e.g., from the Cauchy-Schwarz inequality.
This concludes the estimate of $\|\psi\|_{L^1}$ in the $d=2$ case.\\

Consider now $d=3$. 
The Fourier transform $\hat\psi(k)$ reads, for each $k\in \Lambda^*$ (defined in \eqref{eq_def_Lambda*}):
\begin{align}
\hat \psi(k) &= \frac{1}{L^{6}}\sum_{p,q\in\Lambda^*} \frac{1}{m^2 + \theta(p)}
\cdot \frac{1}{m^2 + \theta(q)}
\cdot\Big[\frac{1}{m^2+\theta(k-p-q)} 
- \frac{1}{m^2 + \theta(-p-q)}\Big]\nonumber\\
&= 
-\int_0^1 d\alpha\, \frac{1}{L^{6}}\sum_{p,q\in\Lambda^*} \frac{1}{m^2 + \theta(p)}
\cdot \frac{1}{m^2 + \theta(q)}\nonumber\\
&\hspace{2.5cm}
\cdot\frac{4\pi \epsilon^{-1} \sum_{i=1}^3 k_i\sin \Big(\frac{\pi\epsilon(\alpha k_i - p_i -q_i)}{2}\Big)\cos \Big(\frac{\pi\epsilon(\alpha k_i - p_i -q_i)}{2}\Big)}{\big[m^2 + \theta(\alpha k - p - q)\big]^2}. 
\label{eq_hatpsidef}
\end{align}
In the following it is convenient to extend the above definition to $k\in\R^3$.
Using convexity to bound the sines and comparing sums to integrals, 
then bounding these integrals by integrals on $\R^d$ and changing variables $p,q\to 2p/m\pi,2q/m\pi$ as before, 
we find that $\hat\psi(k)$ satisfies:
\begin{equation}
\forall k\in\R^3,\qquad
|\hat\psi(k)|
\leq 
c\hat f(2k/m\pi),
\end{equation}
with:
\begin{equation}
\hat f(k) 
:= 
\|k\|_2\int_0^1 d\alpha\int_{\R^{6}}dp\,dq\,  \frac{1}{1+\|p\|_2^2}
\cdot\frac{1}{1+\|q\|_2^2}
\cdot \frac{\|\alpha k - p - q\|_2}{\big[1+\|\alpha k - p - q\|_2^2\big]^2}.
\end{equation}
Let us estimate $\hat f$. 
We claim that, for each $\tilde k\in\R^3$,
\begin{equation}
\int_{\R^3} \frac{1}{1+\|p\|^2}\frac{\|\tilde k - p\|_2}{\big[1+\|\tilde k-p\|_2^2\big]^2}\, dp
= 
\int_{\R^3} \frac{1}{1+\|\tilde k+ p\|^2}\frac{\|p\|_2}{\big[1+\|p\|_2^2\big]^2}\, dp 
\leq 
\frac{c}{1+\|\tilde k\|_2^2}.\label{eq_bound_derivative_varphi_term}
\end{equation}
Indeed, to see this, split the integral between the regions $\{\|p\|_2\leq \|\tilde k\|_2/2\}$, $\{\|p\|_2\geq 3\|\tilde k\|_2/2\}$, and the complement of these. 
In the first two regions, $\|\tilde k+p\|_2\geq \|\tilde k\|_2/2$. 
In the complement, bound $\|p\|_2$ by $\|\tilde k\|_2$ to find:
\begin{align}
\int_{\|\tilde k\|_2/2\leq\|p\|_2\leq 3\|\tilde k\|_2/2}\frac{1}{1+\|\tilde k+ p\|^2}\frac{\|p\|_2}{\big[1+\|p\|_2^2\big]^2} \, dp
&\leq 
\frac{c\|\tilde k\|_2^3}{\big[1+\|\tilde k\|_2^2\big]^2}\int_{1/2}^{3/2} \frac{1}{1+\|\tilde k \|_2^2 u^2}\, du \nonumber\\
&\leq \frac{c}{1+\|\tilde k\|_2^2}.
\end{align}
Using the above bound with $\tilde k =\alpha k - q$ yields:
\begin{align}
\hat f(k)
&\leq 
c\|k\|_2\int_0^1d\alpha\int_{\R^3}\frac{1}{1+\|q\|_2^2}\cdot \frac{1}{1+\|\alpha k -q\|_2^2}\, dq\nonumber\\
& =:\|k\|_2  \int_0^1 d\alpha\,  I(\alpha k). \label{eq_def_I_appendix}
\end{align}
For each $\alpha\in[0,1]$, $I(\alpha k)$ reads:
\begin{align}
I(\alpha k) 
&= 
c\int_{0}^\infty\int_{0}^\infty ds_1\, ds_2 \int_{\R^3} \exp\Big[-s_1\big(1+\|q\|_2^2\big) - s_2\big(1+\|\alpha k -q\|_2^2\big)\Big]\, dq\nonumber\\
& \leq 
c \int_{0}^\infty\int_{0}^\infty ds_1\, ds_2\,  e^{-s_1-s_2}\int_{\R^3}\exp\Big[-(s_1+s_2)\Big\|q-\frac{s_2 \alpha k}{s_1+s_2}\Big\|^2_2 - \frac{s_1s_2 \|\alpha k\|^2_2}{s_1+s_2}\Big]\, dq\nonumber\\
&  \leq 
c \int_{0}^\infty\int_{0}^\infty ds_1\, ds_2\,  \frac{1}{(s_1+s_2)^{3/2}}e^{-s_1-s_2}\exp\Big[- \frac{s_1s_2 \|\alpha k\|^2_2}{s_1+s_2}\Big].
\label{eq_bound_I_appendix}
\end{align}
Changing variable from $s_1$ to $s=s_1/s_2$ yields:
\begin{align}
I(\alpha k) 
&\leq
c \int_0^\infty\int_0^\infty \exp\Big[-s_2\Big(\frac{s\|\alpha k\|_2^2}{s+1}+s+1\Big)\Big]\frac{1}{\sqrt{s_2}(s+1)^{3/2}}\, ds\, ds_2\nnb
& \leq 
c \int_0^\infty \frac{ds}{(s+1) \sqrt{(s+1)^2 + s\|\alpha k\|_2^2}}\nonumber\\
& \leq 
\frac{c}{1+\|\alpha k\|_2}
\label{eq_bound_I_appendix2}
\end{align}
using
\begin{equation}
\sqrt{(s+1)^2+s\|\alpha k\|_2^2}\geq s+1 + \sqrt{s}\|\alpha k\|_2\geq \sqrt{s}(1+\|\alpha k\|_2).
\end{equation}
Thus:
\begin{equation}
\hat f(k)\leq c \log (1+\|k\|_2)\quad \Rightarrow\quad |\hat\psi(k)|\leq c \log\Big(1+\frac{2\|k\|_2}{m\pi}\Big).
\end{equation}
We now prove the bound on $\|C_{(m)}\star\psi\|_{L^1\cap L^2}$. 
One has, first for the $L^2$ norm:
\begin{equation}
\|C_{(m)}\star\psi\|_{L^2}
\leq 
c \bigg(\int_{\R^3}\frac{1}{\big(m^2+\frac{4\|k\|_2^2}{\pi^2}\big)^2}\log\Big(1+\frac{2\|k\|_2}{m\pi}\Big)^2dk\bigg)^{1/2}
\leq 
\frac{c}{m^{1/2}}.
\end{equation}
By similar computations as in the bound on $\hat\psi$, one can prove $|\partial^\alpha\hat \psi(k)| \leq c_\alpha m^{-|\alpha|} (1+\frac{4\|k\|_2^2}{m^2\pi^2})^{-|\alpha|/2}$ for any multi-index $\alpha\in\N^{3}$ and any $k\in\R^3$. 
In addition, for each $j\in\{1,\dots,d\}$ and each $F\in\R^{\Lambda_{\epsilon,L}}$,
\begin{align}
L^2\Big[1-\cos\Big(\frac{2\pi x_j}{L}\Big)\Big]F(x) 
&=
\frac{L^2}{2L^d}\sum_{k\in\Lambda^*}e^{ik\cdot x}\hat F(k)\Big[2-e^{\frac{2i\pi }{L} {\bf e}_j\cdot x} - e^{-\frac{2i\pi }{L} {\bf e}_j\cdot x}\Big]\nnb
&= 
-\frac{1}{2L^d}\sum_{k\in\Lambda^*}e^{ik\cdot x}L^2\Big[\hat F\Big(k+\frac{2\pi}{L}{\bf e}_j\Big) + \hat F\Big(k-\frac{2\pi}{L}{\bf e}_j\Big) - 2\hat F(k)\Big].\nnb
&= 
-\frac{1}{L^d}\sum_{k\in\Lambda^*}e^{ik\cdot x} 
\int_{0}^1 (1-\alpha)\partial^2_j \hat F\Big(k+\frac{2\pi\alpha}{L}{\bf e}_j\Big)\, d\alpha. 
\end{align}
where, for the last line above, the Fourier transform $\hat F$ is extended from $\Lambda^*$ to $\R^3$ as in \eqref{eq_hatpsidef} for $\hat\psi$.
Letting $U(x):= L^2\big(d-\sum_{j=1}^d\cos(2\pi x_j/L)\big)$ for $x\in\Lambda_{\epsilon,L}$, 
noticing that $\|U^{-1}\|_{L^2}\leq c$ and using the Cauchy-Schwarz inequality, this implies
\begin{align}
  \|C_{(m)}\star\psi\|_{L^1}
  &\leq \|U^{-1}\|_{L^2} \|U\cdot C_{(m)}\star \psi\|_{L^2}
    \nnb
  &\leq \|U^{-1}\|_{L^2} 
    \bigg[
    \bigg(\frac{1}{L^{d}}\sum_{k\in\Lambda^*}\int_0^1 \Delta \big(\hat C_{(m)}\hat\psi\big)\Big(k+\frac{2\pi \alpha{\bf e}_j}{L}\Big)\, d\alpha\bigg)^2\,
    \bigg]^{1/2}
        \nnb  
  &\leq c
      \bigg[
      \int_{\R^3}\Big(\frac{1}{\big(m^2+\frac{4\|k\|_2^2}{\pi^2}\big)^2}+ \frac{1}{\big(m^2+\frac{4\|k\|_2^2}{\pi^2}\big)^3}\Big)\Big(1+\log\Big(1+\frac{2\|k\|_2}{m\pi}\Big)^2\Big)\, dk \nnb
  &\hspace{4cm}+
    \int_{\R^3}\Big(\frac{1}{\big(m^2+\frac{4\|k\|_2^2}{\pi^2}\big)^3}+ \frac{1}{\big(m^2+\frac{4\|k\|_2^2}{\pi^2}\big)^4}\Big)\, dk\bigg]^{1/2}\nonumber\\
  &\leq \frac{c}{m^{1/2}} + \frac{c}{m^{3/2}} + \frac{c}{m^{5/2}} 
    \leq 
    \frac{c}{m^{1/2}} + \frac{c}{m^{5/2}},
\end{align}
where the Laplace operator $\Delta$ on the second line above acts on the Fourier variable.
\end{proof}
\begin{lemma}\label{lemm_k=0_Q_term}
The term $\|C_{(m)}(C_{(m)}^2\star C_{(m)}^2)\|_{L^1}$ satisfies:
\begin{equation}
\|C_{(m)}(C_{(m)}^2\star C_{(m)}^2)\|_{L^1}
\leq 
\frac{c}{m^4}{\bf 1}_{d=2} + \frac{c}{m}{\bf 1}_{d=3}.
\end{equation}
\begin{proof}
The proof is similar to that of Lemma~\ref{lemm_psi_term}, 
so we only give a sketch. Using the Fourier representation, one has:
\begin{equation}
\|C_{(m)}(C_{(m)}^2\star C_{(m)}^2)\|_{L^1} 
= 
\frac{1}{L^{d}}\sum_{k\in\Lambda^*} \frac{1}{m^2+\theta(k)}\Big(\frac{1}{L^{d}}\sum_{p\in\Lambda^*}\frac{1}{m^2+\theta(p)}\frac{1}{m^2+\theta(k-p)}\Big)^2.
\end{equation}
The idea is again to compare sums to integrals, and then change variables to factor out the mass dependence: 
\begin{equation}
\|C_{(m)}(C_{(m)}^2\star C_{(m)}^2)\|_{L^1} 
\leq 
\frac{cm^{3d}}{m^{10}} \int_{\R^d} \frac{dk}{1+\|k\|_2^2}
\Big(
\int_{\R^d}dp\,  \frac{1}{1+ \|p\|_2^2}\frac{1}{1+\|k-p\|^2_2}
\Big)^2.
\end{equation}
It remains to check that the last integral is finite.
Notice however that the integral in the square above is exactly $I(k)$, 
as defined in \eqref{eq_def_I_appendix} in the three-dimensional case. 
It follows from \eqref{eq_bound_I_appendix2} that $I(k)\leq C/(1+\|k\|_2)$ in dimension $3$. 
The same arguments also work in dimension $2$, which concludes the proof.
\end{proof}
\end{lemma}
\subsection{Differences of the counterterms}
In this section (and as in Section \ref{sec:suscept-phi4d}), 
$C_t$ is the covariance $C_t = (-\Delta^\epsilon + m^2+1/t)^{-1}$ for $t>0$.
\begin{lemma}\label{lemm_counterterm}
In dimension $d\in\{2,3\}$, the difference of the counterterms at parameters $t>0$ and $+\infty$ satisfies:
\begin{equation}
\begin{split}
&0
\leq 
C_\infty(0) - C_t(0) =:\eta_t \leq c\log\Big(1+\frac{1}{m^2t}\Big){\bf 1}_{d=2} + c m\Big(\sqrt{1 + \frac{1}{tm^2}} - 1\Big){\bf 1}_{d=3},\\
&0
\leq 
\|C^3_\infty\|_{L^1} - \|C^3_t\|_{L^1} := \gamma_t 
\leq 
\frac{c}{m^2(m^2t+1)} {\bf 1}_{d=2} 
+ 
c\log\Big(1+\frac{1}{m^2t}\Big){\bf 1}_{d=3}.
\end{split}
\end{equation}
\end{lemma}
\begin{proof}
The proof is similar to and uses notations of the proof of Lemma \ref{lemm_psi_term}. 
The key argument consists in passing from discrete sums to integrals, using the relation $f(1)-f(0)=\int_0^1 f'(t)\, dt$. 
For $C_\infty(0) - C_t(0)$, one has:
\begin{align}
C_\infty(0) - C_t(0) 
&= 
\frac{1}{L^d}\sum_{k\in \Lambda^*} \Big[\frac{1}{m^2 + \theta(k)} - \frac{1}{m^2 + \frac{1}{t}+ \theta(k)}\Big]\nonumber\\
&= 
\frac{1}{t}\int_0^1 d\alpha\, \frac{1}{L^d}\sum_{k\in\Lambda^*} \frac{1}{\big(m^2 + \frac{\alpha}{t} + \theta(k)\big)^2}\nonumber\\
&\leq 
\frac{c}{t}\int_0^1d\alpha \int_{\R^d}\frac{dk}{\big(m^2+\frac{\alpha}{t} + \frac{4\|k\|_2^2}{\pi^2}\big)^2}.
\end{align}
Changing variables from $k$ to $p=2k(\pi \sqrt{m^2+\alpha/t})^{-1}$ yields:
\begin{align}
C_\infty(0) - C_t(0) 
&\leq 
\int_0^1 \frac{1}{t\big(m^2 + \frac{\alpha}{t}\big)^{2-\frac{d}{2}}}\,  d\alpha 
\times \int_{\R^d} \frac{1}{(1+\|p\|_2^2)^2}\, dp\nonumber\\
&\leq 
c \int_0^1 \frac{1}{t\big(m^2 + \frac{\alpha}{t}\big)^{2-\frac{d}{2}}} \, d\alpha\nonumber \\
&= 
c\log\Big(1+\frac{1}{m^2t}\Big){\bf 1}_{d=2} + c m\Big(\sqrt{1 + \frac{1}{tm^2}} - 1\Big){\bf 1}_{d=3}.
\end{align}
For the $C^3$ counterterm $\gamma_t$, one has similarly:
\begin{align}
\|C^3_\infty\|_{L^1} - \|C^3_t\|_{L^1} 
&= 
\frac{1}{L^{2d}}\sum_{p,q \in \Lambda^*}\bigg[\frac{1}{m^2 + \theta(p)}\frac{1}{m^2 + \theta(q)}\frac{1}{m^2 + \theta(p+q)}\nonumber\\ 
&\hspace{3cm}- \frac{1}{m^2+\frac{1}{t}+ \theta(p)}\frac{1}{m^2+\frac{1}{t}+ \theta(q)}\frac{1}{m^2+\frac{1}{t}+ \theta(p+q)} \bigg]\nonumber\\
&= 
\frac{3}{t}\int_0^1 d\alpha\, \frac{1}{L^{2d}}\sum_{p,q \in \Lambda^*}\frac{1}{m^2+\frac{\alpha}{t}+ \theta(p)}\frac{1}{\big(m^2+\frac{\alpha}{t}+ \theta(q)\big)^2}\frac{1}{m^2+\frac{\alpha}{t}+ \theta(p+q)}.
\end{align}
Bounding $\theta$ by convexity and comparing discrete sums to integrals, then changing variables
from $p,q$ to $2p(\pi \sqrt{m^2+\alpha/t})^{-1}$ and $ 2q(\pi \sqrt{m^2+\alpha/t})^{-1}$ yields:
\begin{align}
\|C^3_\infty\|_{L^1} - \|C^3_t\|_{L^1} 
&\leq 
\frac{3c}{t}\int_0^1d\alpha \int_{\R^{2d}}\frac{1}{m^2+\frac{\alpha}{t}+ \frac{4\|p\|_2^2}{\pi^2}}\frac{1}{\big(m^2+\frac{\alpha}{t}+ \frac{4\|q\|_2^2}{\pi^2}\big)^2}\frac{1}{m^2+\frac{\alpha}{t}+\frac{4\|p+q\|_2^2}{\pi^2}}\, dp\,dq\nonumber\\
&= 
\int_0^1 d\alpha\, \frac{3c\pi^{2d}}{t\big(m^2 +\frac{\alpha}{t}\big)^{4-d}} \int_{\R^6} \frac{1}{1+\|p\|_2^2}\frac{1}{\big(1+\|q\|_2^2\big)^2}\frac{1}{1+\|p+q\|_2^2}\, dp\,dq.
\end{align}
Since the integral over $p,q$ is bounded, therefore:
\begin{align}
\|C^3_\infty\|_{L^1} - \|C^3_t\|_{L^1} 
&\leq 
c\int_0^1 \frac{1}{t\big(m^2 +\frac{\alpha}{t}\big)^{4-d}}\, d\alpha \nonumber\\
&\leq \frac{c}{m^2(m^2t+1)} {\bf 1}_{d=2} 
+ 
c\log\Big(1+\frac{1}{m^2t}\Big){\bf 1}_{d=3}.
\end{align}
\end{proof}

\section*{Acknowledgements}

We thank T.\ Bodineau and J.\ Ding for helpful discussions, as well as G.\ Ferrando for help with the Fourier computations.
This work was supported by the European Research Council under the European Union's Horizon 2020 research and innovation programme
(grant agreement No.~851682 SPINRG). R.B.\ also acknowledges the hospitality of the Department of Mathematics at McGill University
where part of this work was carried out.

\bibliography{all}
\bibliographystyle{plain}
\end{document}